\documentclass[conference]{IEEEtran}
\IEEEoverridecommandlockouts

\usepackage{amsmath}
\usepackage{amssymb}
\usepackage{amsfonts}
\usepackage{amsthm}
\usepackage{pythonhighlight}
\usepackage{algorithmic}
\usepackage{hyperref}
\usepackage{graphicx}
\usepackage{textcomp}
\usepackage{amsmath}
\usepackage{amssymb}
\usepackage{qcircuit}
\usepackage{xcolor}
\usepackage{tikz}
\usetikzlibrary{arrows,shapes,shapes.geometric}

\newtheorem{corollary}{Corollary}
\newtheorem{theorem}{Theorem}

\begin{document}

\title{Comparison of Superconducting NISQ Architectures}

\author{\IEEEauthorblockN{Benjamin Rempfer}
\IEEEauthorblockA{\textit{Lincoln Laboratory} \\
\textit{Massachusetts Institute of Technology}\\
Lexington, MA USA \\
Benjamin.Rempfer@ll.mit.edu}
\and
\IEEEauthorblockN{Kevin M. Obenland}
\IEEEauthorblockA{\textit{Lincoln Laboratory} \\
\textit{Massachusetts Institute of Technology}\\
Lexington, MA USA \\
Kevin.Obenland@ll.mit.edu}
\thanks{DISTRIBUTION STATEMENT A. Approved for public release. Distribution is unlimited.

This material is based upon work supported by the Under Secretary of Defense for
Research and Engineering under Air Force Contract No. FA8702-15-D-0001. Any opinions,
findings, conclusions or recommendations expressed in this material are those of the
author(s) and do not necessarily reflect the views of the Under Secretary of Defense
for Research and Engineering.

\textcopyright 2024 Massachusetts Institute of Technology.


Delivered to the U.S. Government with Unlimited Rights, as defined in DFARS Part
252.227-7013 or 7014 (Feb 2014). Notwithstanding any copyright notice, U.S. Government
rights in this work are defined by DFARS 252.227-7013 or DFARS 252.227-7014 as detailed
above. Use of this work other than as specifically authorized by the U.S. Government
may violate any copyrights that exist in this work.}
}

\maketitle
\thispagestyle{plain}
\pagestyle{plain}

\begin{abstract}
    Advances in quantum hardware have begun the noisy intermediate-scale quantum (NISQ) computing era. A pressing question is: what architectures are best suited to take advantage of this new regime of quantum machines? We study various superconducting architectures including Google's Sycamore, IBM's Heavy-Hex, Rigetti's Aspen and Ankaa in addition to a proposed architecture we call bus next-nearest neighbor (busNNN). We evaluate these architectures using benchmarks based on the quantum approximate optimization algorithm (QAOA) which can solve certain quadratic unconstrained binary optimization (QUBO) problems. We also study compilation tools that target these architectures, which use either general heuristic or deterministic methods to map circuits onto a target topology defined by an architecture.
\end{abstract}


\section{Introduction}
\label{prelim}

Many commercial and government entities are placing significant emphasis on the research and development of quantum computers. Machines available today are deemed  Noisy Intermediate Scale Quantum (NISQ) as a reference to the fact that they are both noisy and modest in scale. Much research pertaining to benchmarking and demonstrating quantum devices today is limited to small-scale architectures with a few qubits. This research is critical to reducing noise and producing higher fidelity qubits \cite{flux}, but it does not address issues related to the design of future architectures. Only recently have companies begun to run experiments with dozens or even a hundred qubits \cite{googleSupremacy} \cite{IBMUtility}. We aim to address the problem of which architectures are best suited to take full advantage of precisely this intermediate scale of dozens of qubits and beyond. This paper restricts its focus to superconducting devices implementing circuits without error correction. There are a multitude of considerations that must be made when comparing NISQ superconducting architectures. In order to fully understand the design space we introduce some preliminary definitions and nomenclature typically associated with superconducting quantum devices. 

Superconducting devices have a predefined set of quantum gates, called its gateset, which can be executed natively within a set amount of time (referred to as the gate's runtime).  Quantum circuits are commonly written as an ordered set of instructions (gates) that allow arbitrary single-qubit gates $g_1\in SU(2)$ as well as entangling gates between arbitrary pairs of qubits $g_2\in SU(4)/SU(2)\otimes SU(2)\simeq SO(4)$. Real devices cannot enact arbitrary entanglement in this way. Thus, we define the set of pairs of qubits on which entangling interactions are executable on some quantum device as its coupling map (sometimes called its topology). Together, a quantum computer's topology and gateset define its \textit{architecture}. Running a quantum circuit on a specific superconducting architecture requires translating the base circuit to a different quantum circuit (that represents the same unitary) so as to adhere to the restrictions of a specific device.

The process of translating a quantum circuit to run on a specific architecture is known as compiling. Finding an optimal compilation solution is NP-hard \cite{depth}, and thus heuristic techniques are typically employed. Note that we do not consider the lower-level problem of translating quantum gates into sequences of microwave pulses (together with compiling called transpiling). Various heuristics and some limited exact methods exist to tackle the compilation problem for general quantum circuits. For this reason, resource analysis of algorithms applied to various device architectures is important when developing new quantum algorithms or considering the implementation of existing algorithms. To this end, two prominent metrics for comparing quantum algorithms compiled to separate architectures are entangling (typically two-qubit) gate count and circuit depth. Two-qubit gate count is an important metric because these gates are typically the lowest fidelity. The circuit depth $m$ is the number of sequential layers of gates in the total circuit. Single-qubit gates are often neglected in depth calculations as they can be executed in parallel across separate qubits and take less time to execute than entangling gates. The depth of a given circuit can be found as the longest path in the directed acyclic dependency graph it generates. Reducing the number of two-qubits gates and the depth of the circuit both have a direct impact on the fidelity and execution time of the circuit. 

It is worth noting that while depth can be a useful metric, a better measure of the performance of a circuit is
\textit{scheduled time}. The scheduled time of a quantum circuit, on a given quantum computer, is the summation a circuit's critical path weighted with each gate's runtime. A circuit's critical path is the longest ordered set of gates that cannot be run in parallel on a specific quantum computer. For most architectures this corresponds directly to a circuit's depth, where qubits represent the only dependency in the machine. However, this is not always the case; machines could contain  components (such as a bus described in Section \hyperref[arch]{II}) that impose additional structural dependencies on the execution of gates. Differences in runtime and gateset can further distance a circuit's scheduled time from its depth.

Clearly, there are numerous trade-offs and a large design space to consider in analyzing near-term quantum architectures. In order to accurately compare NISQ devices one must take coupling maps, gatesets, compiling tools, and quantum algorithms into account. This study has two main contributions to present. The first is an evaluation of current and proposed superconducting qubit architectures. The second is a comparison of a representative set of compilation tools and techniques. In Figure \hyperref[flowchart]{(1)} we illustrate the compilation process used in our work. We start with a problem instance and a definition of the device architecture. The problem is then compiled to the architecture, which gives a circuit that can be analyzed to determine the scheduled time and two-qubit gate count. All circuits are treated with the same custom optimization pass as to control for differences in commercial implementations.

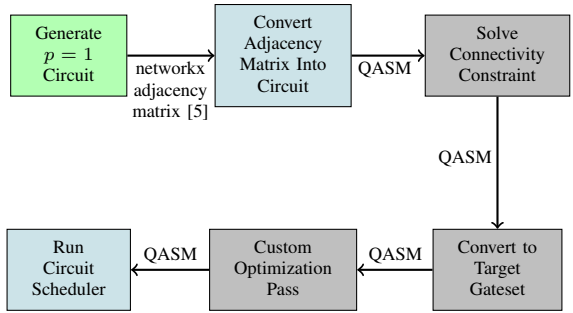
\begin{figure}
\label{flowchart}
    \centering
    \begin{tikzpicture}[
  rect/.style={draw, rectangle},
  or/.style={fill=orange!30},
  bl/.style={fill=blue!10!teal!20},
  gr/.style={fill=green!30},
  gy/.style={fill=black!25!white},
  x=1.5cm,
  y=1.5cm,
  scale=0.95
  ]
  
  \node[rect,gr] (0) at (0,4) {\scriptsize \begin{tabular}{c}
                                Generate \\
                                $p=1$ \\
                                Circuit \\
                                \end{tabular}};
  \node[rect,bl] (1) at (2,4) {\scriptsize \begin{tabular}{c}
                                Convert \\
                                Adjacency \\
                                Matrix Into \\
                                Circuit \\
                                \end{tabular}};
  \node[rect,gy] (2) at (4,4) {\scriptsize \begin{tabular}{c}
                                Solve \\
                                Connectivity \\
                                Constraint \\
                                \end{tabular}};
  \node[rect,gy] (3) at (4,2) {\scriptsize \begin{tabular}{c}
                                Convert to \\
                                Target \\
                                Gateset \\
                                \end{tabular}};
  
  \node[rect,gy] (4) at (2,2) {\scriptsize \begin{tabular}{c}
                                Custom  \\
                                Optimization \\
                                Pass  \\
                                \end{tabular}};
  \node[rect,bl] (5) at (0,2) {\scriptsize \begin{tabular}{c}
                                Run \\
                                Circuit \\
                                Scheduler \\
                                \end{tabular}};
  
  \draw[->, thick] (0) -- (1);
  \draw[->, thick] (1) -- (2);
  \draw[->, thick] (2) -- (3);
  \draw[->, thick] (3) -- (4);
  \draw[->, thick] (4) -- (5);
  
  \node[anchor=north] at (0.95,4.05) {\scriptsize \begin{tabular}{c}
                                networkx \\ 
                                adjacency \\
                                matrix \cite{networkx} \\
                                \end{tabular}};
                                
  \node[anchor=north] at (2.95,4.03) {\scriptsize QASM};
                                
  \node[anchor=north] at (3.7,3.2) {\scriptsize QASM};
                                
  \node[anchor=north] at (3.05,2.3) {\scriptsize QASM};
                                
  \node[anchor=north] at (0.95,2.3) {\scriptsize QASM};
  
\end{tikzpicture}
    \caption{Flowchart showing how a problem instance and architecture are used to produce a compiled circuit which is then converted into a benchmarkable metric. From the final scheduled circuit we extract our chosen metrics.}
\end{figure}

The structure of the paper is as follows; we first describe various existing coupling maps in section \hyperref[arch]{(II)}. In Section \hyperref[qaoa]{(III)} we detail the specific adjacency matrices that specify the Quantum Approximate Optimization Algorithm (QAOA) benchmark circuits. In Section \hyperref[compile]{(IV)}, we give a number of compilation tools that we use and compare. We present the results of our study in Section \hyperref[dsp]{(V)} and finally conclude and discuss future research in Section \hyperref[conclusion]{(VI)}.
\section{Superconducting Qubit Topologies}
\label{arch}

\subsection{Planar}

There are a number of similar, but distinct, superconducting quantum device architectures in development. The three main commercial architectures that we will focus on are Google's Sycamore, Rigetti's Aspen, and IBM's heavy-hex. At the time of writing, the largest publicly announced Google chip is Bristlecone which contains $72$ qubits in a grid-like configuration. Rigetti's $3$rd generation chip uses the Aspen architecture consisting of four octagons with qubits at each vertex positioned in a $4\times2$ grid with two connections between adjacent octagons for a total of $80$ qubits. Recently, Rigetti announced that subsequent generations will no longer utilize this device architecture and instead employ a $2$D grid and be called Ankaa \hyperref[aspen]{\cite{aspen}}. Finally, IBM devices are defined on a ``heavy-hex" architecture \hyperref[heavyhex]{\cite{heavyhex}}. The largest IBM device deployed on the cloud today is $133$ qubits with announcements promising many more. Heavy-hex is made up of tiled hexagons with qubits occupying each vertex and side. For completeness, we also include two generic coupling maps in our analysis: a one-dimensional line and a two-dimensional grid (similar to Rigetti's announced system). A depiction of the three main planar architectures is shown in Figure~\hyperref[commercial]{(2)}.

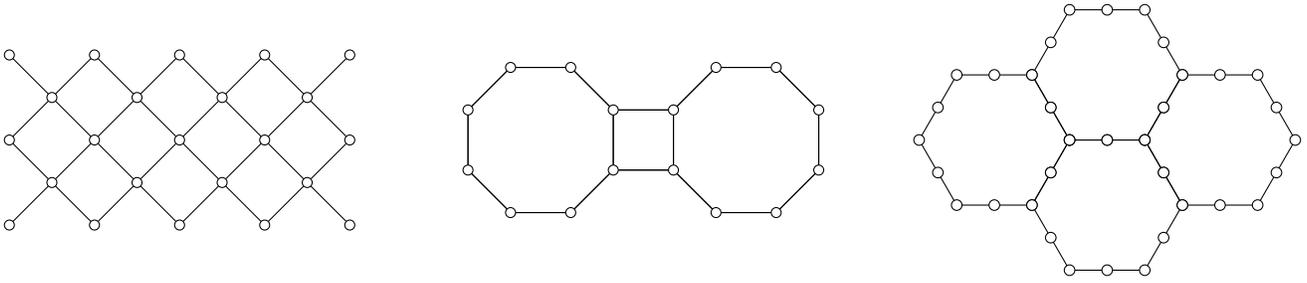
\begin{figure*}[!htb]
\label{commercial}
    \minipage{0.32\textwidth}
        \centering
        \begin{tikzpicture}[auto, node distance = 2cm, main node/.style = {circle, draw, scale = 0.4pt}]
	\node[main node] (0) {};
	\node[main node] (1) [below right of=0] {};
	\node[main node] (2) [below left of=1] {};
	\node[main node] (3) [below right of=2] {};
	\node[main node] (4) [below left of=3] {};
	\node[main node] (5) [above right of=1] {};
	\node[main node] (6) [below right of=1] {};
	\node[main node] (7) [below right of=3] {};
	\node[main node] (8) [below right of=5] {};
	\node[main node] (9) [below right of=6] {};
	\node[main node] (10) [above right of=8] {};
	\node[main node] (11) [below right of=8] {};
	\node[main node] (12) [below right of=9] {};
	\node[main node] (13) [below right of=10] {};
	\node[main node] (14) [below right of=11] {};
	\node[main node] (15) [above right of=13] {};
	\node[main node] (16) [below right of=13] {};
	\node[main node] (17) [below right of=14] {};
	\node[main node] (18) [below right of=15] {};
	\node[main node] (19) [below right of=16] {};
	\node[main node] (20) [above right of=18] {};
	\node[main node] (21) [below right of=18] {};
	\node[main node] (22) [below right of=19] {};
	
	\begin{scope}[-]
		\path
		   	(0) edge (1)
	    	(1) edge (5)
	    	(1) edge (2)
	    	(7) edge (3)
	    	(3) edge (2)
	   	    (6) edge (3)
	    	(4) edge (3)
	    	(1) edge (6)
	    	(8) edge (6)
	    	(5) edge (8)
    		(7) edge (9)
    		(8) edge (10)
    		(6) edge (9)
    		(11) edge (8)
    		(11) edge (9)
    		(9) edge (12)
    		(11) edge (13)
    		(10) edge (13)
    		(13) edge (15)
    		(12) edge (14)
    		(13) edge (16)
    		(11) edge (14)
    		(16) edge (14)
    		(14) edge (17)
    		(16) edge (18)
    		(15) edge  (18)
    		(18) edge (20) 
    		(17) edge (19)
    		(16) edge (19)
    		(18) edge (21)
    		(21) edge (19)
    		(19) edge (22);
    \end{scope}
\end{tikzpicture}
    \endminipage\hfill
    \minipage{0.32\textwidth}
        \centering
        \begin{tikzpicture}[auto, node distance = 2cm, main node/.style={circle, draw, scale = 0.4pt}]
	\node[main node] (0) {};
	\node[main node] (1) [right of=0] {};
	\node[main node] (2) [below left of=0] {};
	\node[main node] (3) [below right of=1] {};
	\node[main node] (4) [below of=2] {};
	\node[main node] (5) [below of=3] {};
	\node[main node] (6) [below right of=4] {};
	\node[main node] (7) [below left of=5] {};
	\node[main node] (8) [right of=3] {};
	\node[main node] (9) [right of=5] {};
	\node[main node] (10) [above right of=8] {};
	\node[main node] (11) [below right of=9] {};
	\node[main node] (12) [right of=10] {};
	\node[main node] (13) [right of=11] {};
	\node[main node] (14) [below right of=12] {};
	\node[main node] (15) [above right of=13] {};
	
	\path[every node/.style={font=\sffamily\small}]
	(0) edge node [right] {} (1)
	(1) edge node [left] {} (0)
	(0) edge node [right] {} (2)
	(2) edge node [left] {} (0)
	(2) edge node [right] {} (4)
	(4) edge node [left] {} (2)
	(4) edge node [right] {} (6)
	(6) edge node [left] {} (4)
	(1) edge node [right] {} (3)
	(3) edge node [left] {} (1)
	(3) edge node [right] {} (5)
	(5) edge node [left] {} (3)
	(5) edge node [right] {} (7)
	(7) edge node [left] {} (5)
	(6) edge node [right] {} (7)
	(7) edge node [left] {} (6)
	(8) edge node [right] {} (10)
	(10) edge node [left] {} (8)
	(8) edge node [right] {} (9)
	(9) edge node [left] {} (8)
	(9) edge node [right] {} (11)
	(11) edge node [left] {} (9)
	(11) edge node [right] {} (13)
	(13) edge node [left] {} (11)
	(13) edge node [right] {} (15)
	(15) edge node [left] {} (13)
	(14) edge node [right] {} (15)
	(15) edge node [left] {} (14)
	(12) edge node [right] {} (14)
	(14) edge node [left] {} (12)
	(10) edge node [right] {} (12)
	(12) edge node [left] {} (10)
	(3) edge node [right] {} (8)
	(8) edge node [left] {} (3)
	(5) edge node [right] {} (9)
	(9) edge node [left] {} (5); 
\end{tikzpicture}   
    \endminipage\hfill
    \minipage{0.32\textwidth}
        \centering
        \begin{tikzpicture}
  	\begin{scope}[every node/.style={anchor=west, regular polygon, regular polygon sides=6, draw, minimum width=2cm, outer sep=0}, transform shape]
    	\node (A) {};
    	\node (B) at (A.corner 1) {};
		\node (C) at (B.corner 5) {};
    	\node (D) at (A.corner 5) {};
   		\foreach \hex in {A,...,D}
    		{
      			\foreach \corn in {1,...,6}
			{
       			\draw[fill=white] (\hex.corner \corn) circle (2pt); 
			}
			\foreach \side in {1,...,6}
			{
				\draw[fill=white] (\hex.side \side) circle (2pt);
			}
    		}
	\end{scope}
\end{tikzpicture}
    \endminipage\hfill
    \centering
    \caption{Commercially available architectures. From left to right: a $23$ qubit example of Goolge's Sycamore architecture, a $16$ qubit subgraph of a Rigetti Aspen chip, and four tiled heavy-hexagons illustrating IBM's architectures.} 
\end{figure*}

The connectivity of a topology is defined as the average number of edges per node. All commercial architectures are planar graphs with average connectivities ranging from $\approx2.27$ for IBM Eagle to $\approx3.36$ for Google Bristlecone. Intuitively, higher connectivity leads to lower depth circuits. However, the amount of parallelism present in the circuit will have an impact on the performance. As will be discussed in Section \hyperref[compile]{(IV)}, in general, a logarithmic depth overhead $\Omega(\log n)$ is necessarily incurred when utilizing coupling maps with finite degree \hyperref[depth]{\cite{depth}}. However, this overhead is only theoretically achievable. In practice, this lower-bound is difficult to reach.

\subsection{Non-Planar}

In addition to the planar coupling maps defined above, we introduce two architectures based on non-planar coupling maps (see Figure~\hyperref[nonplanar]{(3)} for an illustration of these two architectures). These non-planar graphs gain connectivity but are more complex to fabricate than the planar topologies. First is the layered Sycamore chip, which is defined as two Sycamore chips connected via transverse vertical links.
This increases the average number of edges per qubit from $\approx 3.36$ to $\approx 4.36$.

\begin{figure}
\label{nonplanar}
    \centering
        \minipage{0.48\textwidth}
        \centering
        \begin{tikzpicture}[auto, node distance = 2.5cm, main node/.style = {circle, draw, scale = 0.4pt}]
	\node[main node] (0) {};
	\node[main node] (1) [below right of=0] {};
	\node[main node] (2) [below left of=1] {};
	\node[main node] (3) [below right of=2] {};
	\node[main node] (4) [below left of=3] {};
	\node[main node] (5) [above right of=1] {};
	\node[main node] (6) [below right of=1] {};
	\node[main node] (7) [below right of=3] {};
	\node[main node] (8) [below right of=5] {};
	\node[main node] (9) [below right of=6] {};
	\node[main node] (10) [above right of=8] {};
	\node[main node] (11) [below right of=8] {};
	\node[main node] (12) [below right of=9] {};
	\node[main node] (13) [below right of=10] {};
	\node[main node] (14) [below right of=11] {};
	\node[main node] (15) [above right of=13] {};
	\node[main node] (16) [below right of=13] {};
	\node[main node] (17) [below right of=14] {};
	\node[main node] (18) [below right of=15] {};
	\node[main node] (19) [below right of=16] {};
	\node[main node] (20) [above right of=18] {};
	\node[main node] (21) [below right of=18] {};
	\node[main node] (22) [below right of=19] {};
	\node[main node, draw=none] (46) [above right of=0] {};
	\node[main node, draw=none] (47) [above left of=46] {};
	\node[main node, fill=black!25!white] (23) [below of=47] {};
	\node[main node, fill=black!25!white] (24) [below right of=23] {};
	\node[main node, fill=black!25!white] (25) [below left of=24] {};
	\node[main node, fill=black!25!white] (26) [below right of=25] {};
	\node[main node, fill=black!25!white] (27) [below left of=26] {};
	\node[main node, fill=black!25!white] (28) [above right of=24] {};
	\node[main node, fill=black!25!white] (29) [below right of=24] {};
	\node[main node, fill=black!25!white] (30) [below right of=26] {};
	\node[main node, fill=black!25!white] (31) [below right of=28] {};
	\node[main node, fill=black!25!white] (32) [below right of=29] {};
	\node[main node, fill=black!25!white] (33) [above right of=31] {};
	\node[main node, fill=black!25!white] (34) [below right of=31] {};
	\node[main node, fill=black!25!white] (35) [below right of=32] {};
	\node[main node, fill=black!25!white] (36) [below right of=33] {};
	\node[main node, fill=black!25!white] (37) [below right of=34] {};
	\node[main node, fill=black!25!white] (38) [above right of=36] {};
	\node[main node, fill=black!25!white] (39) [below right of=36] {};
	\node[main node, fill=black!25!white] (40) [below right of=37] {};
	\node[main node, fill=black!25!white] (41) [below right of=38] {};
	\node[main node, fill=black!25!white] (42) [below right of=39] {};
	\node[main node, fill=black!25!white] (43) [above right of=41] {};
	\node[main node, fill=black!25!white] (44) [below right of=41] {};
	\node[main node, fill=black!25!white] (45) [below right of=42] {};
	\node[main node, draw=none] (48) [below of=27] {};
	\begin{scope}[-]
        \path[every node/.style={font=\sffamily\small}, white]
            (0) edge (1)
            (1) edge (5)
            (1) edge (2)
            (7) edge (3)
            (3) edge (2)
            (6) edge (3)
            (4) edge (3)
   		    (1) edge (6)
            (8) edge (6)
            (5) edge (8)
            (7) edge (9)
            (8) edge (10)
            (6) edge (9)
            (11) edge (8)
            (11) edge (9)
            (9) edge (12)
            (11) edge (13)
            (10) edge (13)
            (13) edge (15)
            (12) edge (14)
            (13) edge (16)
            (11) edge (14)
            (16) edge (14)
            (14) edge (17)
            (16) edge (18)
            (15) edge  (18)
            (18) edge (20) 
            (17) edge (19)
            (16) edge (19)
            (18) edge (21)
            (21) edge (19)
            (19) edge (22);
        \path[every node/.style={font=\sffamily\small}]
    	    (0) edge (1)
    	    (1) edge (5)
    	    (1) edge (2)
    	    (7) edge (3)
    	    (3) edge (2)
    	    (6) edge (3)
    	    (4) edge (3)
   		    (1) edge (6)
    	    (8) edge (6)
    	    (5) edge (8)
    	    (7) edge (9)
    	    (8) edge (10)
    	    (6) edge (9)
    	    (11) edge (8)
    	    (11) edge (9)
    	    (9) edge (12)
   	        (11) edge (13)
    	    (10) edge (13)
    	    (13) edge (15)
    	    (12) edge (14)
    	    (13) edge (16)
    	    (11) edge (14)
    	    (16) edge (14)
    	    (14) edge (17)
    	    (16) edge (18)
    	    (15) edge  (18)
    	    (18) edge (20) 
    	    (17) edge (19)
    	    (16) edge (19)
    	    (18) edge (21)
    	    (21) edge (19)
    	    (19) edge (22);
        \path[every node/.style={font=\sffamily\small}, white]
		    (23) edge (24)
    	    (24) edge (28)
    	    (24) edge (25)
    	    (30) edge (26)
    	    (26) edge (25)
    	    (29) edge (26)
    	    (27) edge (26)
    	    (24) edge (29)
    	    (31) edge (29)
    	    (28) edge (31)
    	    (30) edge (32)
   	        (31) edge (33)
    	    (29) edge (32)
    	    (34) edge (31)
    	    (34) edge (32)
    	    (32) edge (35)
    	    (34) edge (36)
    	    (33) edge (36)
    	    (36) edge (38)
    	    (35) edge (37)
            (36) edge (39)
    	    (34) edge (37)
    	    (39) edge (37)
    	    (37) edge (40)
    	    (39) edge (41)
    	    (38) edge (41)
    	    (41) edge (43) 
    	    (40) edge (42)
    	    (39) edge (42)
    	    (41) edge (44)
    	    (44) edge (42)
    	    (42) edge (45);
        \path[every node/.style={font=\sffamily\small}, black]
		    (23) edge (24)
    	    (24) edge (28)
    	    (24) edge (25)
    	    (30) edge (26)
    	    (26) edge (25)
    	    (29) edge (26)
    	    (27) edge (26)
    	    (24) edge (29)
    	    (31) edge (29)
    	    (28) edge (31)
    	    (30) edge (32)
   		    (31) edge (33)
    	    (29) edge (32)
    	    (34) edge (31)
    	    (34) edge (32)
    	    (32) edge (35)
    	    (34) edge (36)
    	    (33) edge (36)
    	    (36) edge (38)
    	    (35) edge (37)
   		    (36) edge (39)
    	    (34) edge (37)
    	    (39) edge (37)
    	    (37) edge (40)
    	    (39) edge (41)
    	    (38) edge  (41)
    	    (41) edge (43) 
    	    (40) edge (42)
    	    (39) edge (42)
    	    (41) edge (44)
    	    (44) edge (42)
    	    (42) edge (45);
        \path[every node/.style={font=\sffamily\small}, black!65!white]
		    (0) edge (23)
    	    (4) edge (27)
    	    (20) edge (43)
    	    (22) edge (45)
    	    (2) edge (25)
    	    (21) edge (44)
    	    (10) edge (33)
    	    (11) edge (34)
    	    (12) edge (35)
    	    (5) edge (28)
    	    (6) edge (29)
    	    (7) edge (30)
   	        (15) edge (38)
    	    (16) edge (39)
    	    (17) edge (40)
    	    (1) edge (24)
    	    (3) edge (26)
    	    (18) edge (41)
    	    (19) edge (42)
    	    (8) edge (31)
    	    (9) edge (32)
    	    (13) edge (36)
    	    (14) edge (37);
    \end{scope}
\end{tikzpicture}
    \endminipage\hfill
        \minipage{0.48\textwidth}
        \centering
        \begin{tikzpicture}[auto, node distance = 1.75cm, main node/.style = {circle, draw, scale = 0.4pt}, box/.style = {rectangle, draw, fill=black!15!white, minimum width=0.4pt, minimum height=24mm}]

	\node[main node, fill=black!30!white] (0) {};
	\node[main node, fill=black!30!white] (1) [below of=0] {};
	\node[main node, fill=black!30!white] (2) [below of=1] {};
	\node[main node, fill=black!30!white] (3) [below of=2] {};
	
	\node[main node, draw=none] (4) [above right of=0] {};
	
	\node[main node, draw=none] (5) [below right of=4] {};
	\node[main node, draw=none] (6) [above of=5] {};
	\node[main node, draw=none] (7) [below of=5] {};
	\node[main node, draw=none] (8) [below of=6] {};
	\node[main node, draw=none] (9) [below of=7] {};

	\node[box] (10) [below of=6] {};
	
	\node[main node, draw=none] (11) [above right of=5] {};
	
	\node[main node, fill=black!30!white] (22) [below right of=11] {};
	\node[main node, fill=black!30!white] (23) [below of=22] {};
	\node[main node, fill=black!30!white] (24) [below of=23] {};
	\node[main node, fill=black!30!white] (25) [below of=24] {};
	
    \node[main node, draw=none] (26) [above right of=22] {};	
	
	\node[main node, fill=black!30!white] (12) [below right of=26] {};
	\node[main node, fill=black!30!white] (13) [below of=12] {};
	\node[main node, fill=black!30!white] (14) [below of=13] {};
	\node[main node, fill=black!30!white] (15) [below of=14] {};
	
	\node[main node, draw=none] (16) [above right of=12] {};
	
	\node[main node, draw=none] (17) [below right of=16] {};
	\node[main node, draw=none] (18) [above of=17] {};
	\node[main node, draw=none] (19) [below of=17] {};
	\node[main node, draw=none] (20) [below of=19] {};
	\node[main node, draw=none] (21) [below of=20] {};

	\node[box] (27) [below of=18] {};
	
	\node[main node, draw=none] (28) [above right of=17] {};
	
	\node[main node, fill=black!30!white] (29) [below right of=28] {};
	\node[main node, fill=black!30!white] (30) [below of=29] {};
	\node[main node, fill=black!30!white] (31) [below of=30] {};
	\node[main node, fill=black!30!white] (32) [below of=31] {};
	
    \node[main node, draw=none] (33) [above right of=29] {};	
	
	\node[main node, fill=black!30!white] (34) [below right of=33] {};
	\node[main node, fill=black!30!white] (35) [below of=34] {};
	\node[main node, fill=black!30!white] (36) [below of=35] {};
	\node[main node, fill=black!30!white] (37) [below of=36] {};
	
	\node[main node, draw=none] (38) [above right of=34] {};
	
	\node[main node, draw=none] (39) [below right of=38] {};
	\node[main node, draw=none] (40) [above of=39] {};
	\node[main node, draw=none] (41) [below of=39] {};
	\node[main node, draw=none] (42) [below of=41] {};
	\node[main node, draw=none] (43) [below of=42] {};

	\node[box] (44) [below of=40] {};
	
	\node[main node, draw=none] (45) [above right of=39] {};
	\node[main node, fill=black!30!white] (46) [below right of=45] {};
	\node[main node, fill=black!30!white] (47) [below of=46] {};
	\node[main node, fill=black!30!white] (48) [below of=47] {};
	\node[main node, fill=black!30!white] (49) [below of=48] {};

	\begin{scope}[-]
		\path
	 	   	(22) edge (12)
	     	(23) edge (13)
	     	(24) edge (14)
	     	(25) edge (15)
	 	   	(12) edge (22)
	     	(13) edge (23)
	     	(14) edge (24)
	     	(15) edge (25)
	 	   	(29) edge (34)
	     	(30) edge (35)
	     	(31) edge (36)
	     	(32) edge (37)
	 	   	(34) edge (29)
	     	(35) edge (30)
	     	(36) edge (31)
	     	(37) edge (32)
	     	(0) edge (10)
	     	(1) edge (10)
	     	(2) edge (10)
	     	(3) edge (10)
	     	(10) edge (0)
	     	(10) edge (1)
	     	(10) edge (2)
	     	(10) edge (3)
	        (22) edge (10)
	        (23) edge (10)
	        (24) edge (10)
	        (25) edge (10)
	        (10) edge (22)
	        (10) edge (23)
	        (10) edge (24)
	        (10) edge (25)
	        (12) edge (27)
	        (13) edge (27)
	        (14) edge (27)
	        (15) edge (27)
	        (27) edge (12)
	        (27) edge (13)
	        (27) edge (14)
	        (27) edge (15)
	        (29) edge (27)
	        (30) edge (27)
	        (31) edge (27)
	        (32) edge (27)
	        (27) edge (29)
	        (27) edge (30)
	        (27) edge (31)
	        (27) edge (32)
	        (34) edge (44)
	        (35) edge (44)
	        (36) edge (44)
	        (37) edge (44)
	        (44) edge (34)
	        (44) edge (35)
	        (44) edge (36)
	        (44) edge (37)
	        (34) edge (44)
	        (35) edge (44)
	        (36) edge (44)
	        (37) edge (44)
	        (46) edge (44)
	        (47) edge (44)
	        (48) edge (44)
	        (49) edge (44)
	        (44) edge (46)
	        (44) edge (47)
	        (44) edge (48)
	        (44) edge (49);
    \end{scope}
\end{tikzpicture}
    \endminipage\hfill
    \vspace{10pt}
    \caption{A layered configuration of $23$ qubit sycamore chips (top) and a graphical representation of the busNNN topology (bottom). In the busNNN architecture, each rectangle represents a bus that can be dynamically reconfigured to connect any pair of qubits connected to it. The number of qubits connected to
    a bus is a parameter of the bus definition, and we assume that a bus can only execute a single two-qubit gate at a time (although this could conceivably be altered).}
\end{figure}

Our final coupling map utilizes a (qubit) \textit{bus} \hyperref[bus]{\cite{bus}} \hyperref[busComponent]{\cite{busComponent}}. This bus allows arbitrary two-qubit gates between any pair of qubits connected to it. However, each entangling gate must be sequentialized. The bus increases connectivity but comes with the trade off of added hardware complexity and reduced parallelism. There is a natural optimization between the number of qubits connected via a bus $|B|$ and the total number of buses $B$. Larger buses reduce the number of swaps required but also reduce the parallelism of a circuit. Our base-line design utilizes buses with $|B|=8$ qubits, which are connected via point-to-point connections (an architecture we call bus next-nearest neighbor (busNNN) as illustrated in Figure \hyperref[nonplanar]{(3)}. These architectures are listed in Table \hyperref[tab1]{(I)}.

\begin{table*}
\label{tab1}
    \centering
    \begin{tabular}{ |p{2.25cm}|p{1cm}|p{3.25cm}|p{1.75cm}|p{2cm}|  }
        \hline
        \multicolumn{5}{|c|}{All Considered Coupling Maps} \\
        \hline
        \centering
        Topology & Qubits & Connectivity & Type & Existence\\
        \hline
        Sycamore & $72$ & $\approx3.36$ & Planar & Commercial \\ 
        Aspen & $80$ & $2.5$ & Planar & Commercial \\
        Eagle (heavy-hex) & $127$ & $\approx2.27$ & Planar & Commercial \\
        Line & $n$ & $2(n-1)/n$ & Planar & Hypothetical \\
        $2$D Grid & $N^2$ & $4(N-1)/N$ & Planar & Hypothetical \\
        Layered & $144$ & $\approx4.36$ & Non-Planar & Hypothetical \\
        busNNN & $B$ & $|B|(|B| B-1)/2$ & Non-Planar & Hypothetical \\
        \hline
    \end{tabular}
    \vspace{10pt}
    \caption{Coupling maps for the architectures considered in this study. The size of each device instance is either fixed or varied to accommodate the number of qubits in a circuit.}
\end{table*}
\section{QAOA Problem Instances}
\label{qaoa}

In our study of superconducting architectures, it is important to benchmark relevant quantum algorithms. One such algorithm is the Quantum Approximate Optimization Algorithm (QAOA) which is itself a 2-local Hamiltonian simulation. Its circuit thus has qualitatively similar structure to other important quantum algorithms including Hamiltonian simulation for Heisenberg, XY, Ising, and Fermi-Hubbard models \hyperref[IsingNP]{\cite{IsingNP}} \hyperref[tutorialQUBO]{\cite{tutorialQUBO}} \hyperref[NPQUBO]{\cite{NPQUBO}}. 

\begin{table*}
\label{tab2}
    \centering
    \begin{tabular}{ |p{4.5cm}||p{3cm}|p{5.7cm}| }
        \hline
        \multicolumn{3}{|c|}{All Considered Problem Instances} \\
        \hline
        \centering
        Model Graph & Description & Applications \\
        \hline
        3 regular & degree 3 & maxcut \\
        12 regular & degree 12 & maxcut \\
        Watts-Strogatz (WS) & high clustering & small-world, community structure, etc. \\
        Barab{\'a}si-Albert (BA) & power law growth & social networks, price models, etc. \\
        Sherrington-Kirkpatrick (SK) & fully connected & Ising and spin-glass \\
        Erd{\"o}s R{\'e}nyi (ER) & ``average" graph & unstructured problems \\
        \hline
    \end{tabular}
    \vspace{10pt}
    \caption{Graph types considered in this study. The analysis is performed using $p=1$ QAOA implementations.}
\end{table*}

QAOA produces approximate solutions for combinatorial optimization problems \hyperref[QAOA]{\cite{QAOA}} \hyperref[boundedQAOA]{\cite{boundedQAOA}} \hyperref[QAOAsupremacy]{\cite{QAOAsupremacy}}. These problems are specified by minimizing an objective function \[C(z)=\sum_{ij}A_{ij}z_iz_j\] over $z\in\{\pm1\}^N$ where $A\in\mathcal{M}(\mathbb{R})_{N\times N}$ is the adjacency matrix for the problem instance encoded as an undirected weighted graph \hyperref[combinatorialOptimization]{\cite{combinatorialOptimization}}. The quantum circuits consists of $p\in\mathbb{Z}^+$ layers with classical processing sandwiching each. In the version of QAOA solving maxcut, all layers consist of the parameterized unitary \[U(\beta,\gamma)=\prod\mbox{exp}(i\gamma ZZ)\prod\mbox{exp}(i\beta X)\] where $\gamma,\beta\in\mathbb{R}$ are the angles to be classically optimized and change from layer to layer. In the limit $p\rightarrow\infty$, QAOA converges to an optimal solution. As every layer is identical, we consider $p=1$ QAOA max-cut circuits for simplicity. It is an active area of research to determine the value of $p$ where QAOA yields quantum advantage in terms of speed and/or solution quality. Near-term devices can only run relatively-low depth circuits and thus stand firmly in the low-$p$ regime \hyperref[relaxRound]{\cite{relaxRound}} \hyperref[recursiveQAOA]{\cite{recursiveQAOA}}. Note that, for the subsequent discussion of problem instance candidates, the weights in each graph contribute nothing to our desired metrics related to compilation and can therefore be ignored. Our study is not of the performance of QAOA, but rather the scheduling and compilation of circuits to target coupling maps, for which we choose QAOA to be representative of 2-local Hamiltonian simulation circuits. We explicitly employ QAOA for max-cut circuits because their qualitative structure is similar to 2-local Hamiltonian simulation circuits and they are straightforward to generate.

\begin{figure}[!htb]
\label{png}
    \minipage{0.16\textwidth}
        \centering
        \begin{tikzpicture}[lineDecorate/.style = {-}, nodeDecorate/.style = {shape = circle, inner sep = 1pt, draw}, scale = 1]
\foreach \nodename/\x/\y in {
  20/1.00000000000000/0.000000000000000,
  1/0.951056516295154/0.309016994374947,
  2/0.809016994374947/0.587785252292473,
  3/0.587785252292473/0.809016994374947,
  4/0.309016994374947/0.951056516295154,
  5/0.000000000000000/1.00000000000000,
  6/-0.309016994374947/0.951056516295154,
  7/-0.587785252292473/0.809016994374947,
  8/-0.809016994374947/0.587785252292473,
  9/-0.951056516295154/0.309016994374947,
  10/-1.00000000000000/0.000000000000000,
  11/-0.951056516295154/-0.309016994374947,
  12/-0.809016994374947/-0.587785252292473,
  13/-0.587785252292473/-0.809016994374947,
  14/-0.309016994374947/-0.951056516295154,
  15/0.000000000000000/-1.00000000000000,
  16/0.309016994374947/-0.951056516295154,
  17/0.587785252292473/-0.809016994374947,
  18/0.809016994374947/-0.587785252292473,
  19/0.951056516295154/-0.309016994374947}
{
  \node (\nodename) at (\x,\y) [nodeDecorate] {};
}
\path
\foreach \startnode/\endnode in {
  1/2, 1/4, 1/5, 1/7, 1/11, 1/19, 1/20, 2/9, 2/18, 3/6, 3/8, 3/17, 4/8, 4/9, 4/14, 4/17,
  4/19, 4/20, 5/7, 5/9, 5/12, 5/14, 5/17, 5/18, 5/19, 6/7, 6/10, 6/14, 6/17, 7/12, 7/14,
  8/20, 9/11, 9/16, 10/13, 10/19, 11/14, 11/16, 11/17, 11/18, 11/19, 11/20, 12/13, 13/14,
  13/16, 13/17, 13/18, 13/20, 14/15, 16/18, 16/20, 18/19, 18/20}
{
  (\startnode) edge[lineDecorate] node {} (\endnode)
};
\end{tikzpicture}
    \endminipage\hfill
    \minipage{0.16\textwidth}
        \centering
        \begin{tikzpicture}[lineDecorate/.style = {-}, nodeDecorate/.style = {shape = circle, inner sep = 1pt, draw}, scale = 1, noNode/.style = {shape = circle, inner sep = 1pt, draw = none}]
\foreach \nodename/\x/\y in {
  20/1.00000000000000/0.000000000000000,
  1/0.951056516295154/0.309016994374947,
  2/0.809016994374947/0.587785252292473,
  3/0.587785252292473/0.809016994374947,
  4/0.309016994374947/0.951056516295154,
  5/0.000000000000000/1.00000000000000,
  6/-0.309016994374947/0.951056516295154,
  7/-0.587785252292473/0.809016994374947,
  8/-0.809016994374947/0.587785252292473,
  9/-0.951056516295154/0.309016994374947,
  10/-1.00000000000000/0.000000000000000,
  11/-0.951056516295154/-0.309016994374947,
  12/-0.809016994374947/-0.587785252292473,
  13/-0.587785252292473/-0.809016994374947,
  14/-0.309016994374947/-0.951056516295154,
  15/0.000000000000000/-1.00000000000000,
  16/0.309016994374947/-0.951056516295154,
  17/0.587785252292473/-0.809016994374947,
  18/0.809016994374947/-0.587785252292473,
  19/0.951056516295154/-0.309016994374947}
{
\node (\nodename) at (\x,\y) [nodeDecorate] {};
}

\foreach \nodename/\x/\y in {
  40/0.5/0,
  21/0.47552825814758/0.15450849718747,
  22/0.404/0.29,
  23/0.29/0.404,
  24/0.15/0.48,
  25/0/0.5,
  26/-0.15/0.48,
  27/-0.29/0.404,
  28/-0.404/0.29,
  29/-0.48/0.15,
  30/-1/0,
  31/-0.48/-0.15,
  32/-0.404/-0.29,
  33/-0.29/-0.404,
  34/-0.15/-0.48,
  35/0/-1,
  36/0.15/-0.48,
  37/0.29/-0.404,
  38/0.404/-0.29,
  39/0.48/-0.15}
{
  \node (\nodename) at (\x, \y) [noNode] {};
}

\path

\foreach \startnode/\endnode in {
    20/2, 1/3, 2/4, 3/5, 4/6, 5/7, 6/8, 7/9, 8/10, 9/11, 10/12, 11/13, 12/14, 13/15, 14/16, 15/17, 16/18, 17/19, 18/20, 19/1, 20/6, 4/10, 12/19, 12/16}
    {
    (\startnode) edge[out=90,in=90,relative] node {} (\endnode)
    }

\foreach \startnode/\endnode in {
    20/1, 1/2, 3/4, 4/5, 5/6, 6/7, 7/8, 8/9, 9/10, 10/11, 12/13, 13/14, 15/16, 16/17, 17/18, 19/20}
    {
    (\startnode) edge node {} (\endnode)
    };

\end{tikzpicture}
    \endminipage\hfill
    \minipage{0.16\textwidth}
        \centering
        \begin{tikzpicture}[lineDecorate/.style = {-}, nodeDecorate/.style = {shape = circle, inner sep = 1pt, draw}, scale = 1]
\foreach \nodename/\x/\y in {
  20/0/0,
  1/0.13/0,
  2/0.217/0.283,
  3/-0.17/0.283,
  4/-0.2/-0.27,
  5/0.27/-0.23,
  6/.5/-.5,
  7/0.63/-0.4,
  8/0.7/0.4,
  9/0.93/0.53,
  10/0.87/0.3,
  11/0.7/0.67,
  12/-0.37/-0.47,
  13/-0.3/-0.63,
  14/-0.6/-0.67,
  15/-0.53/0.23,
  16/-0.8/0.47,
  17/-0.6/0.4,
  18/-0.83/0.83,
  19/-0.67/0.97
  }
{
  \node (\nodename) at (\x,\y) [nodeDecorate] {};
}
\path
\foreach \startnode/\endnode in {
  20/1, 20/2, 20/3, 20/4, 20/5, 5/6, 5/7, 2/8, 8/9, 8/10, 8/11, 4/12, 12/13, 12/14, 3/15, 15/16, 15/17, 17/18, 17/19}
{
  (\startnode) edge[lineDecorate] node {} (\endnode)
};
\end{tikzpicture}

    \endminipage\hfill
    \centering
    \caption{On the left is an Erd{\"o}s R{\'e}nyi graph with $D\approx1/2$. The center and right depict a Watts-Strogatz and Barab{\'a}si-Albert graph respectively all with $20$ vertices.} 
\end{figure}
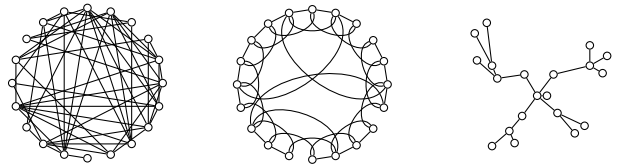

As alluded to above, $2-$local Hamiltonian simulation algorithms are naturally described by an adjacency matrix, with nodes corresponding to variables (qubits) and edges corresponding to terms in the Hamiltonian.
There are a number of choices for graphs to be used as a problem instances including random Erd{\"o}s-R{\'e}nyi (ER) graphs. The ER model selects a graph with $N$ nodes and $M$ edges by choosing from all possible graphs with density $D=2M/N(N+1)$ uniformly at random. We can define multiple unique graphs, leading to unique QAOA circuits, of a given density.

Random graphs, however, may be qualitatively different from graphs corresponding to real problem instances. Therefore we must expand our search to other types of graphs that are perhaps more explicitly linked to real problems in combinatorial optimization. We turn to connected $3$-regular and $12$-regular graphs with unit weight. Random regular graphs are given by placing edges between vertices until every vertex is connected to the specified number of other vertices, in this case $3$ or $12$. We also consider completely-connected Sherrington Kirkpatrick (SK) graphs. This is a specific case of ER graphs with $D=1$. It provides a useful benchmark for the most dense graphs possible. SK graphs can be used to define a type of spin-glass model. 

The final two problem instances we use are Watts-Strogatz (WS) and Barab{\'a}si-Albert (BA) graphs. WS graphs improve the clustering of ER graphs. ER graphs do not generate local clustering or triadic closures (they possess a low clustering coefficient). WS graphs account for clustering while retaining short path lengths. A WS graph is constructed by starting with a ring graph and adding $k=4$ next-nearest neighboring connections. Then one rewires edge $ij$ to $ik$ with probability $p=1/2$. BA graphs address another key shortcoming of the ER model: the scaling of degree. The degree distribution of ER graphs converges to the Poisson distribution, as opposed to the power law which is more realistic. These graphs are built by starting with a star graph with $\lceil n/4+1\rceil$ nodes. Then an additional $\lceil3n/4-1\rceil$ nodes with $\lceil n/4\rceil$ edges are attached to higher degree existing nodes.

We illustrate ER, WS, and BA graphs in Figure \hyperref[png]{(4)}. In Table \hyperref[tab2]{(II)} we summarize the graph types used to create benchmark circuits in our study.
\section{Compiling Tools}
\label{compile}

\subsection{Heuristic}
Similar to Section \hyperref[arch]{(II)} detailing the many existing commercial coupling maps, there are a comparable number of compiling tools. The hardware and algorithm agnostic heuristics we consider are IBM's implementation of SwAp-based BidiREctional search (SABRE) \hyperref[SABRE]{\cite{SABRE}} and the router provided by Tket \hyperref[tket]{\cite{tket}}. In addition to these general-purpose tools we consider an algorithmic-specific tool called 2QAN, which takes advantage of the structure of $2$-local Hamiltonian's for further optimization not possible in the tools based on general heuristics \hyperref[2qan]{\cite{2qan}}. 

The first heuristic compiling tool, SwAp-based BiDiREctional search (SABRE), utilizes a directed acyclic graph representation of the quantum circuit. It is based on two main design principles: prioritizing gates with near-term dependencies, and the novel idea of traversing the reversed circuit to find an initial mapping. (Finding a good initial mapping is a nontrivial component of compilation.) Tket's router uses a time-sliced version of a circuit to construct an initial mapping. In both heuristics a distance measure is used as a cost function. Afterwards, we utilize some of qiskit's optimization passes to clean up any unwanted artifacts of compilation \hyperref[peephole]{\cite{peephole}}.

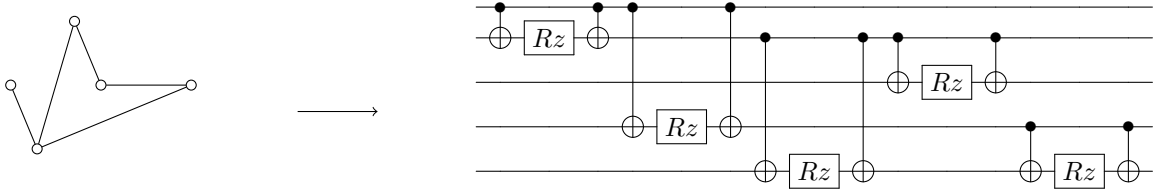
\begin{figure*}[!htb]
\label{circuit}
\noindent
\begin{minipage}{.3\textwidth}
    \begin{tikzpicture}[auto,node distance=3cm,main node/.style={circle,draw,scale=0.4pt}]
    
            \node[main node, color=white] (10) {};
            \node[main node, color=white] (11) [above right of=10] {};
            
			\node[main node] (0) [right of=11]{};
			\node[main node, color=white] (1) [right of=0] {};
			\node[main node] (2) [below left of=0] {};
			\node[main node] (3) [below left of=1] {};
			\node[main node] (4) [right of=3] {};
			\node[main node] (5) [below left of=3] {};
			
			\node[main node, color=white] (6) [above right of=3] {};
			\node[main node, color=white] (7) [below right of=6] {};
			\node[main node, color=white] (8) [above right of=7] {};
			\node[main node, color=white] (9) [below right of=8] {};
			\node[main node, color=white] (12) [below of=6] {};
			\node[main node, color=white] (13) [above right of=12] {};
			\node[main node, color=white] (14) [below right of=13] {};
			\node[main node, color=white] (15) [right of=14] {};
			
		\begin{scope}[-]
			\path
    			(0) edge (5)
    			(0) edge (3)
			    (5) edge (4)
    			(5) edge (2)
			    (3) edge (4);
		\end{scope}
		\begin{scope}[->]
		    \path 
		        (14) edge (15);
		\end{scope}
	\end{tikzpicture}
\end{minipage}
\begin{minipage}{.7\textwidth}
\[\Qcircuit @C=0.5em @R=0.4em {
        && \ctrl{1} & \qw & \ctrl{1} & \ctrl{3} & \qw & \ctrl{3} & \qw & \qw & \qw & \qw & \qw & \qw & \qw & \qw & \qw & \qw \\
        && \targ & \gate{Rz} & \targ & \qw & \qw & \qw & \ctrl{3} & \qw & \ctrl{3} & \ctrl{1} & \qw & \ctrl{1} & \qw & \qw & \qw & \qw \\
        && \qw & \qw & \qw & \qw & \qw & \qw & \qw & \qw & \qw & \targ & \gate{Rz} & \targ & \qw & \qw & \qw & \qw \\
        && \qw & \qw & \qw & \targ & \gate{Rz} & \targ & \qw & \qw & \qw & \qw & \qw & \qw & \ctrl{1} & \qw & \ctrl{1} & \qw \\
        && \qw & \qw & \qw & \qw & \qw & \qw & \targ & \gate{Rz} & \targ & \qw & \qw & \qw & \targ & \gate{Rz} & \targ & \qw \\}
    \]
\end{minipage}
    \caption{On the left is an example problem instance graph. The right depicts a subcircuit of a QAOA instance on $5$ qubits corresponding to that example problem instance. Exp($i\theta ZZ$) operators can be permuted to reduce the depth of the circuit from $5$ time steps to $3$. In general, a 2-local Hamiltonian's operators may not commute but can still be arbitrarily permuted.}
\end{figure*}

SABRE and Tket's heuristic algorithms are powerful. However, they operate at the gate level and therefore cannot take advantage of structural properties of the algorithm. The ability to optimize at the operator-level allows for additional resource reductions. In QAOA it is obvious to see that all entangling operators commute and therefore reordering at the operator level results in significant improvements to the compilation solution. It is also the case, however, that any $2-$local Hamiltonian can be reformulated with effectively commuted entangling operations. This potential optimization is completely invisible to SABRE and Tket. In Figure \hyperref[circuit]{(5)} we show a circuit that can be reduced in depth by 2QAN, but not by the other tools. 

In other words, there are $N_A$ distinct adjacency matrices for the same problem graph $G$ where \[N_A=\frac{n!}{|\mbox{Aut}(G)|}\] and $\mbox{Aut}(G)$ is the group of functions preserving the group structure of $G$  \hyperref[wolframAdjacency]{\cite{wolframAdjacency}}. The quantum circuits generated from these different matrices will be represented via distinct unitaries and are thus inaccessible to the gate-level optimization present in many commercial compilers. 2QAN takes advantage of this added optimization space via a Tabu search \cite{tabu} at the expense of compilation runtime.

\begin{table}
\label{tab3}
    \centering
    \begin{tabular}{ |p{1.5cm}||p{3.25cm}|p{2.75cm}| }
        \hline
        \multicolumn{3}{|c|}{All Considered Compilers} \\
        \hline
        \centering
        Compiler & Description & Provider \\
        \hline
        SABRE & traverses DAG twice & IBM \\
        Tket & time slices circuit & Quantinuum \\
        2QAN & Tabu search & open source \\
        shuffle & low computational overhead & internal implementation \\
        \hline
    \end{tabular}
    \vspace{10pt}
    \caption{All compiling tools considered in this study. These tools provide a representative sample of currently available tools.}
\end{table}

\subsection{Deterministic}

\begin{figure}
\label{swap}    
    \centering
    \begin{tikzpicture}[auto,node distance=2cm,main node/.style={circle,draw,scale=0.8pt}]
			\node[main node] (0) {0};
			\node[main node] (1) [right of=0] {1};
			\node[main node] (2) [right of=1] {2};
			\node[main node] (3) [right of=2] {3};
			\node[main node] (4) [right of=3] {4};
			\node[main node] (5) [right of=4] {5};
			\begin{scope}[-]
				\path
    				(0) edge (1)
    				(1) edge (2)
				(2) edge (3)
    				(3) edge (4)
				(4) edge (5);
			\end{scope}
		\end{tikzpicture}\\
		\begin{tikzpicture}[->,>=stealth',shorten >=1pt,auto,node distance=2cm, main node/.style={circle,draw,scale=0.8pt}]
			\node[main node, color=white] (0) {};
			\node[main node, color=white] (1) [right of=0] {};
			\node[main node, color=white] (2) [right of=1] {};
			\node[main node, color=white] (3) [right of=2] {};
			\node[main node, color=white] (4) [right of=3] {};
			\node[main node, color=white] (5) [right of=4] {};
			\path[every node/.style={font=\sffamily\small}]
			(0) edge node [right] {} (1)
			(1) edge node [right] {} (0)
			(2) edge node [right] {} (3)
			(3) edge node [right] {} (2)
			(4) edge node [right] {} (5)
			(5) edge node [right] {} (4);  
		\end{tikzpicture}\\
		\begin{tikzpicture}[auto,node distance=2cm,main node/.style={circle,draw,scale=0.8pt}]
			\node[main node] (0) {1};
			\node[main node] (1) [right of=0] {0};
			\node[main node] (2) [right of=1] {3};
			\node[main node] (3) [right of=2] {2};
			\node[main node] (4) [right of=3] {5};
			\node[main node] (5) [right of=4] {4};
			\begin{scope}[-]
				\path
    				(0) edge (1)
    				(1) edge (2)
				(2) edge (3)
    				(3) edge (4)
				(4) edge (5);
			\end{scope}
		\end{tikzpicture}\\
		\begin{tikzpicture}[->,>=stealth',shorten >=1pt,auto,node distance=2cm, main node/.style={circle,draw,scale=0.8pt}]
			\node[main node, color=white] (0) {};
			\node[main node, color=white] (1) [right of=0] {};
			\node[main node, color=white] (2) [right of=1] {};
			\node[main node, color=white] (3) [right of=2] {};
			\node[main node, color=white] (4) [right of=3] {};
			\node[main node, color=white] (5) [right of=4] {};
			\path[every node/.style={font=\sffamily\small}]
			(1) edge node [right] {} (2)
			(2) edge node [right] {} (1)
			(3) edge node [right] {} (4)
			(4) edge node [right] {} (3);
		\end{tikzpicture}\\
		\begin{tikzpicture}[auto,node distance=2cm,main node/.style={circle,draw,scale=0.8pt}]
			\node[main node] (0) {1};
			\node[main node] (1) [right of=0] {3};
			\node[main node] (2) [right of=1] {0};
			\node[main node] (3) [right of=2] {5};
			\node[main node] (4) [right of=3] {2};
			\node[main node] (5) [right of=4] {4};
			\begin{scope}[-]
				\path
    				(0) edge (1)
    				(1) edge (2)
				(2) edge (3)
    				(3) edge (4)
				(4) edge (5);
			\end{scope}
		\end{tikzpicture}\\
		\begin{tikzpicture}[->,>=stealth',shorten >=1pt,auto,node distance=2cm, main node/.style={circle,draw,scale=0.8pt}]
			\node[main node, color=white] (0) {};
			\node[main node, color=white] (1) [right of=0] {};
			\node[main node, color=white] (2) [right of=1] {};
			\node[main node, color=white] (3) [right of=2] {};
			\node[main node, color=white] (4) [right of=3] {};
			\node[main node, color=white] (5) [right of=4] {};
			\path[every node/.style={font=\sffamily\small}]
			(0) edge node [right] {} (1)
			(1) edge node [right] {} (0)
			(2) edge node [right] {} (3)
			(3) edge node [right] {} (2)
			(4) edge node [right] {} (5)
			(5) edge node [right] {} (4);
		\end{tikzpicture}\\
		\begin{tikzpicture}[auto,node distance=2cm,main node/.style={circle,draw,scale=0.8pt}]
			\node[main node] (0) {3};
			\node[main node] (1) [right of=0] {1};
			\node[main node] (2) [right of=1] {5};
			\node[main node] (3) [right of=2] {0};
			\node[main node] (4) [right of=3] {4};
			\node[main node] (5) [right of=4] {2};
			\begin{scope}[-]
				\path
    				(0) edge (1)
    				(1) edge (2)
				(2) edge (3)
    				(3) edge (4)
				(4) edge (5);
			\end{scope}
		\end{tikzpicture}\\
		\begin{tikzpicture}[->,>=stealth',shorten >=1pt,auto,node distance=2cm, main node/.style={circle,draw,scale=0.8pt}]
			\node[main node, color=white] (0) {};
			\node[main node, color=white] (1) [right of=0] {};
			\node[main node, color=white] (2) [right of=1] {};
			\node[main node, color=white] (3) [right of=2] {};
			\node[main node, color=white] (4) [right of=3] {};
			\node[main node, color=white] (5) [right of=4] {};
			\path[every node/.style={font=\sffamily\small}]
			(1) edge node [right] {} (2)
			(2) edge node [right] {} (1)
			(3) edge node [right] {} (4)
			(4) edge node [right] {} (3);
		\end{tikzpicture}\\
		\begin{tikzpicture}[auto,node distance=2cm,main node/.style={circle,draw,scale=0.8pt}]
			\node[main node] (0) {3};
			\node[main node] (1) [right of=0] {5};
			\node[main node] (2) [right of=1] {1};
			\node[main node] (3) [right of=2] {4};
			\node[main node] (4) [right of=3] {0};
			\node[main node] (5) [right of=4] {2};
			\begin{scope}[-]
				\path
    				(0) edge (1)
    				(1) edge (2)
				(2) edge (3)
    				(3) edge (4)
				(4) edge (5);
			\end{scope}
		\end{tikzpicture}\\
    \centering
    \caption{A fully shuffled line topology with $n=6$ component qubits. Each row depicts a swap reconfiguration step, with time flowing from top to bottom.}
\end{figure}
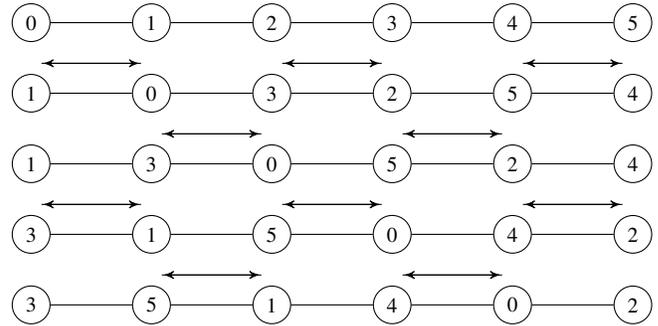

One crucial drawback, with these heuristic algorithms, is that their runtime scales poorly at large problem sizes. Deterministic routing algorithms, on the other hand, can compile very large circuits with predictable, low constant factor scaling. One such algorithm is the full-shuffle which was independently developed by Jin et. al. (called Multi-dimensional All-to-All Permutation and Swaps MAAPS) \hyperref[maaps]{\cite{maaps}} and IBM (implemented in commutative SABRE) \hyperref[commtSABRE]{\cite{commtSABRE}}. These techniques provide scalability to larger system sizes when one has access to patterned coupling maps. The line and $2$D grid strategies have been worked out in the papers above. This allows us to formalize the runtime of shuffle networks, which we restate for completeness in the following theorems.

\begin{theorem}
For the line graph of size $n$, the swap strategy that alternates between two swap layers, one on all odd numbered edges, and one on all even numbered edges reaches full connectivity in $n-2$ layers and is optimal.
\end{theorem}

The proof is given in the references above and intuition is illustrated in Figure \hyperref[swap]{(6)}. Note that a provably optimal full shuffle strategy for a $2$D grid has not yet been discovered. Still, a swap strategy has been developed; it's proof is also contained in the above references. The two-dimensional strategy alternates between applying $N-1$ steps of the line swap strategy to each row individually and swapping all rows by applying exactly two steps of the line swap strategy to each column. The code given in Figure \hyperref[code]{(7)} shows how to generate those layers and store them in the dictionary swaps with keys corresponding to each layer index. We develop an analogous full shuffle strategy for our proposed busNNN architecture. The design of this architecture is motivated by the optimality of its swap strategy, however it is still prudent to test its performance with other compilation algorithms.

\begin{figure}
\label{code}
\begin{center}
\begin{python}
for row in range(N):
  for col in range(row*N,row*N+N):
    if ((((col
      if (row
        swaps[0].append((col,col+1))
      if (row
        swaps[1].append((col,col+1))
    if ((((col
      if (row
        swaps[1].append((col,col+1))
      if (row
        swaps[0].append((col,col+1))
for col in range(N):
  for row in range(N):
    if ((((row+col*N)
      if ((col+row*N+N)<=(N**2-1)):
        swaps[2].append((col+row*N,col+row*N+N))
    if ((((row+col*N)
      if ((col+row*N+N)<=(N**2-1)):
        swaps[3].append((col+row*N,col+row*N+N))
\end{python}
\end{center}
\caption{Python code to generate the swap layers (stored in a dictionary of lists of pairs) needed for a full shuffle of some $N\times N$ 2D grid architecture.}
\end{figure}

\begin{theorem}
For the $2$D grid of size $n=N^2$ the swap strategy of alternating $N-1$ steps of the line swap strategy to each row and swapping rows in parallel reaches full connectivity in no more than $\frac{1}{2}(N-2)(N+1)$ layers.
\end{theorem}

\begin{theorem}
For the busNNN architecture with $B$ buses the swap strategy that alternates between two swap layers, one between adjacent half columns not contained within the same bus, and one between half columns on the same bus except the $B$-th bus, reaches full connectivity in precisely $(4B-5)\lceil\frac{B}{B+1}\rceil$ layers. 
\end{theorem}

\begin{proof}
Without loss of generality label the $i$-th half bus at time $t'$ as $C(i,t')$ where indexing begins at $0$ and $t'$ is evolved by each swap layer. Denote half the qubits connected to a bus as a column and note that each column contains $|B|/2$ qubits. Following the initial swap layer,
\[C(i+1,1)=\left\{
\begin{array}{ll}
    C(i,0) + 1 & \dots (i\neq2B-1)\wedge (i\ \mbox{even}) \\
    C(i,0) - 1 & \dots (i\neq2B-1)\wedge (i\ \mbox{odd}) \\
    C(i,0) & \dots i=2B-1 \\
\end{array} 
\right.\]
Clearly no duplicate combinations of columns have occurred within a bus when comparing $t'=0$ and $t'=1$. This is not the case when comparing $t'>0$ odd and even. However, we can see that skipping every other time step yields a distinct set of pairs of columns aboard each bus. The sequence of columns that arrive at the $0$-th index at time $t=2t'-1$ with $t'>0$ is given by,
\[C(0,t)=\left\{
\begin{array}{ll}
    C(0,2(t-1)) & \dots t\leq B \\
    C(0,2(2B-t)-1) & \dots B<t<2B \\
\end{array} 
\right.\]
We can see that no spatially adjacent columns repeat at any moment in time when attached to the same bus by observing every $C(i,t)$ being identical across indices $i$ up to a phase of exactly one reverse time step. Since these sequences are cyclic they will repeat after $2B-1$ time steps. This, along with the initial mapping, ensures full connectivity between columns, and by extension qubits, is reached in 2B time steps. Transforming back to time being the number of swap layers we see full connectivity is reached in $4B-5$ layers. Still, the case of one bus requires no swaps so we insert a factor of 
\[\left\lceil\frac{B}{B+1}\right\rceil=\left\{
\begin{array}{ll}
    0 & \dots\ B=1 \\
    1 & \dots\ B>1 \\
\end{array} 
\right.\]
which gives the desired result. Code is given to generate this swap strategy in Figure \hyperref[busCode]{(8)}.
\end{proof}

\begin{figure}
\label{busCode}
\begin{center}
\begin{python}
if B > 1:
  for bus in range(1,B):
    for qbit in range(|B|/2*(2*bus-1),|B|*bus):
      swaps[1].append([qbit-|B|/2,qbit])
swaps[0] = [swap for swap in interbus_connections]
\end{python}
\end{center}
\caption{Python code to generate the swap layers (stored in a dictionary of lists of pairs) needed for a full shuffle of our busNNN architecture.}
\end{figure}

\begin{corollary}
The swap strategy for busNNN is optimal in terms of number of swap layers required to reach full connectivity.
\end{corollary}

For completeness, we now address the scaling of each coupling map's required number of swaps to reach full connectivity. The line strategy requires \[L_{SWAP}=\frac{1}{2}(n-1)(n-2).\] The $N\times N$ $2$D grid requires a maximum of
\tiny
\[G_{SWAP}=\left\{
\begin{array}{ll}
    \left\lceil\sqrt{N}\right\rceil\left(\left\lceil\sqrt{N}\ \right\rceil + 1\right)\left\lceil\frac{\left\lceil\sqrt{N}\ \right\rceil}{2}\right\rceil\left\lceil\frac{2\left\lceil\sqrt{N}\ \right\rceil - 3}{4}\right\rceil \dots \left\lceil N\right\rceil \mbox{odd} \\ \\
    \left(\frac{\left\lceil\sqrt{N}\ \right\rceil}{2}\left(\left\lceil\sqrt{N}\ \right\rceil + 1\right)^2+2\left\lceil\sqrt{N}\right\rceil\left\lceil\frac{2\left\lceil\sqrt{N}\ \right\rceil - 3}{4}\right\rceil\right)\left\lceil\frac{\left\lceil\sqrt{N}\ \right\rceil}{2}\right\rceil \dots \mbox{else} \\
\end{array} 
\right. \]
\normalsize
Finally, the busNNN topology requires \[B_{SWAP}=\frac{|B|}{2}\left\lceil\frac{n-|B|}{|B|}\right\rceil\left(\frac{|B|}{2}\left\lceil\frac{n-|B|}{|B|}\right\rceil - 1\right).\] 

The scaling of these swap counts has a direct impact on the number of entangling gates. However, their relation to the depth is more clear in the number of layers as mentioned in theorems $(1-3)$. Note that there are instances where optimization can still be applied to lower this overhead. Such a case is when the entangling gateset contains a sole CNOT gate. Then, in the busNNN architecture, one can match swap gates and entangling gates operating on the same qubits to reduce the total number of two-qubit gates used for the combined operation \hyperref[mikeIke]{\cite{mikeIke}}. It is worth reiterating just how important deterministic algorithms, like the full shuffle, can be. They are easy to compute and thus scale better than heuristic algorithms at larger problem sizes. This transition from heuristic dominance to breakdown may well be within the range of NISQ devices. The main question we should then address is at what regime does this transition occur?
\section{Exploring the Design Space}
\label{dsp}

In this section we present the main results. We first compare all available compilers for a set of problem instances with the same coupling maps. We will utilize the number of entangling gates and scaled scheduled time as our metrics. Scaled scheduled time is the scheduled time when a single qubit gate takes $1$ unit of time and an entangling gate takes $10$. As mentioned before, we compile to the gateset of CNOT plus X and Z single qubit rotations. This is to ensure we are comparing compilers and coupling maps and not the idiosyncrasies of the gatesets of particular architectures. Because our random graphs result in different adjacency matrices, we generate $100$ graph instances for each configuration. Each circuit is compiled on a single Xeon-p8 processor maintained by the MIT LLSC team \hyperref[grid]{\cite{grid}}. The ``baseline" plot represents a fully connected architecture requiring no swaps or compiling passes.

\begin{figure*}[!htb]
\label{busNNNlinegrid}
    \minipage{0.32\textwidth}
        \centering
        \includegraphics[width=0.95\textwidth]{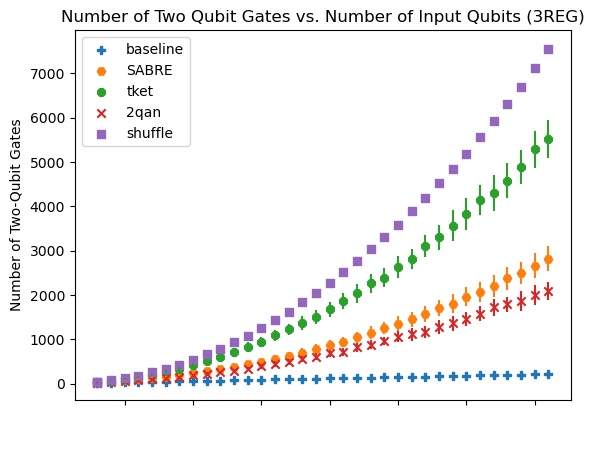}
        \includegraphics[width=0.95\textwidth]{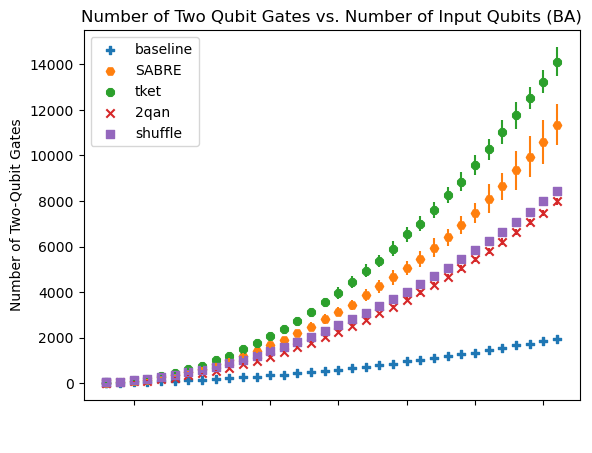}
        \includegraphics[width=0.95\textwidth]{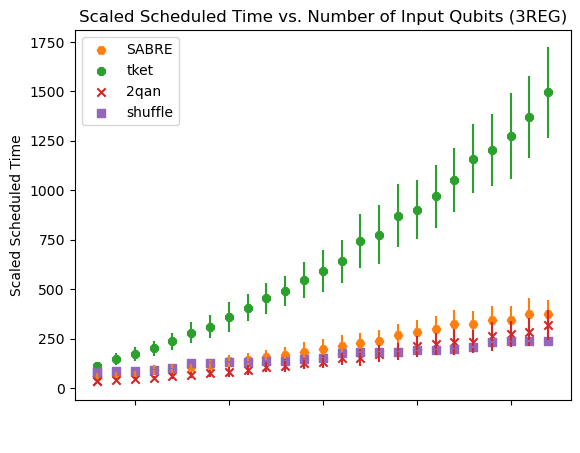}
        \includegraphics[width=0.95\textwidth]{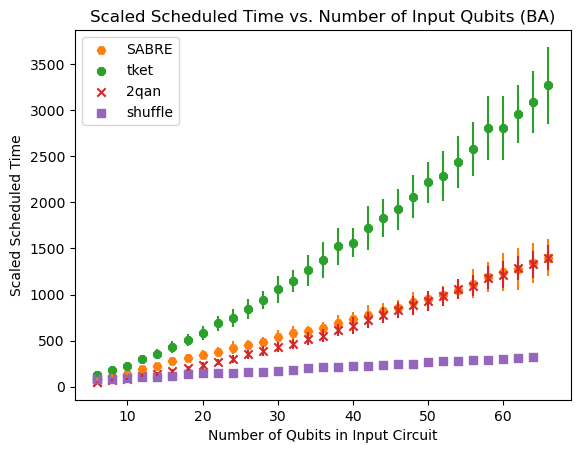}
    \endminipage\hfill
    \minipage{0.32\textwidth}
        \centering
        \includegraphics[width=0.95\textwidth]{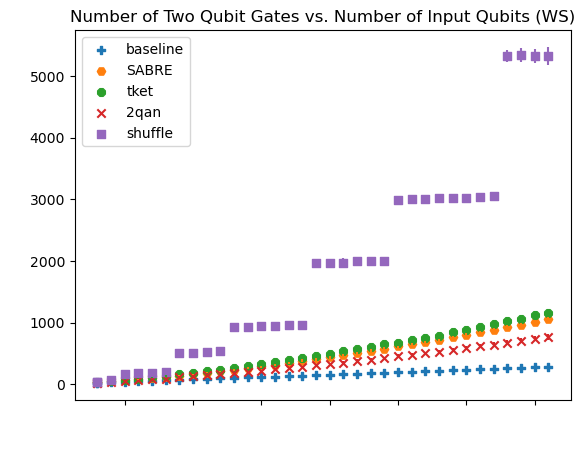}
        \includegraphics[width=0.95\textwidth]{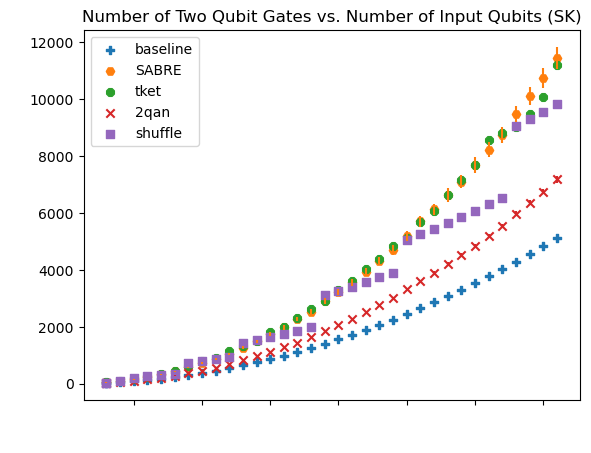}
        \includegraphics[width=0.95\textwidth]{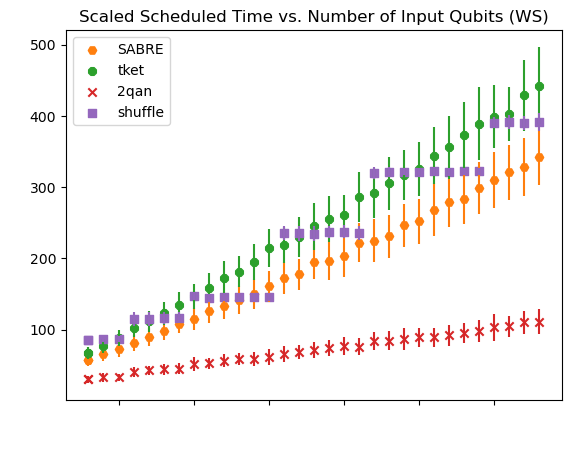}
        \includegraphics[width=0.95\textwidth]{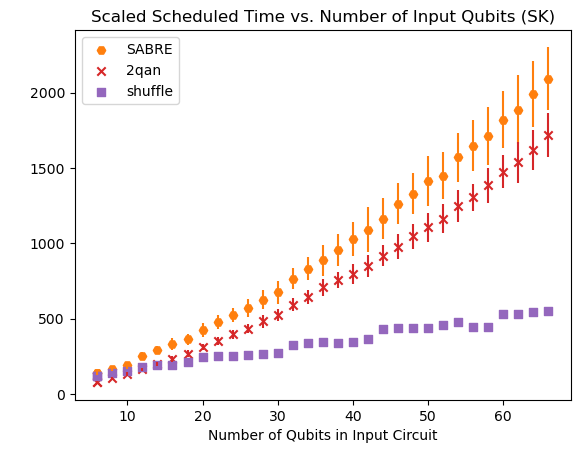}
    \endminipage\hfill
    \minipage{0.32\textwidth}
        \centering
        \includegraphics[width=0.95\textwidth]{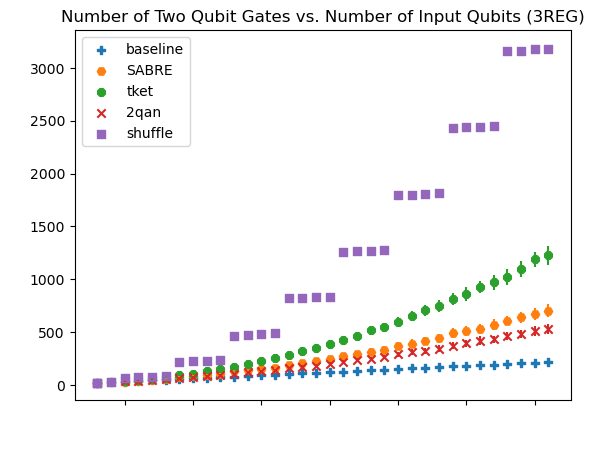}
        \includegraphics[width=0.95\textwidth]{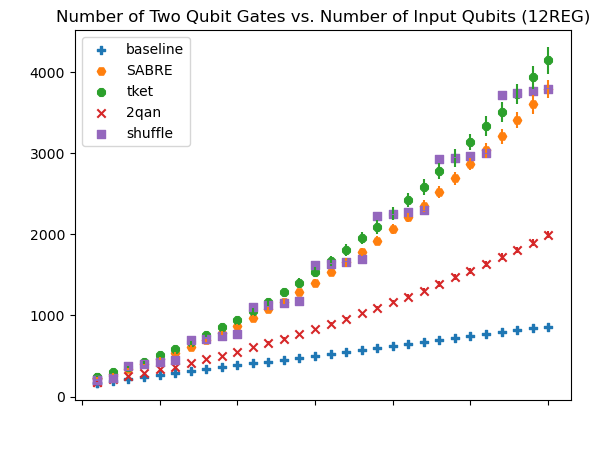}
        \includegraphics[width=0.95\textwidth]{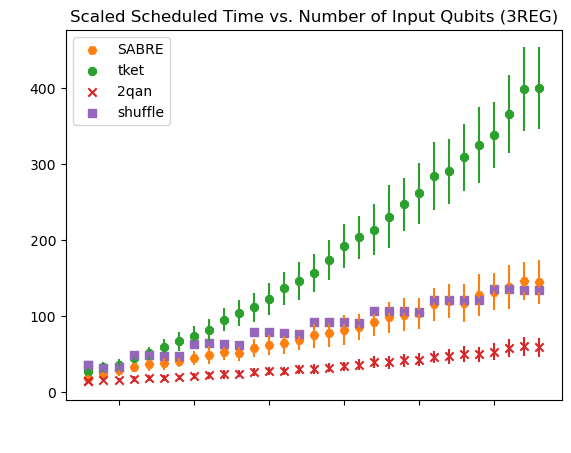}
        \includegraphics[width=0.95\textwidth]{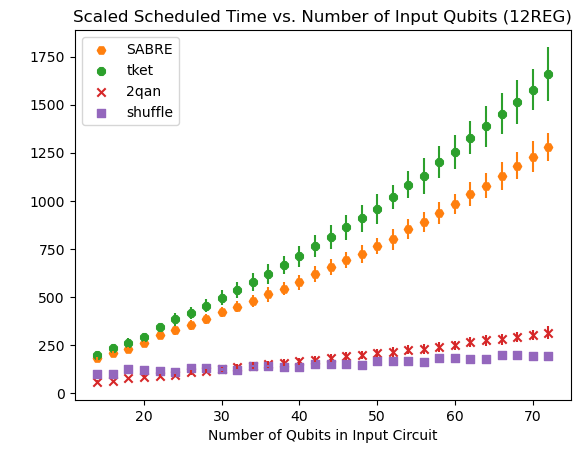}
    \endminipage\hfill
    \caption{The left column contains problem instances compiled on a line-type architecture. The middle and right are 2D grid and busNNN respectively. The type of problem instance is chosen to be representative of a wide variety of scaling. Tket is omitted from the bottom middle graph because its counts are so large that the rest of the data is obscured. Each graph contains data points for $6\leq q\leq72$ qubits}.
\end{figure*}

In Figure \hyperref[busNNNlinegrid]{(9)}, we see that although 2QAN does better in all cases with respect to the number of two-qubit gates. Although, the margin of that difference varies significantly. While shuffle's scaling is deterministic, it cannot compete with any heuristic when applied to a sparse graph such as 3REG. As system size increases the density of regular graphs decreases via the relation $D=d/(n-1)$ where $d$ is the degree of regularity. We can then make some qualitative assertions on the performance of the given compilation tools. \textit{Shuffle becomes better than 2QAN with larger and denser problem instances.} Using scheduled time as a metric yields significantly different results. While the above italicized statement is still true, the threshold where shuffle performs better than all heuristics is within the demonstrated data: convincingly within the NISQ regime. Broadly, SABRE performs better than Tket and takes significantly less compute time to run. 

\begin{figure}[!htb]
\label{timing}
    \centering
    \includegraphics[width=0.375\textwidth]{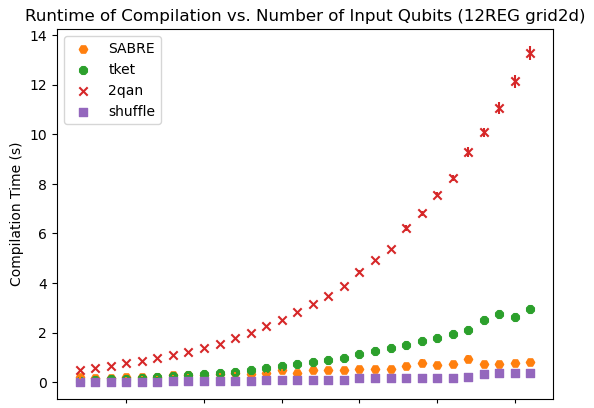}
    \includegraphics[width=0.35\textwidth]{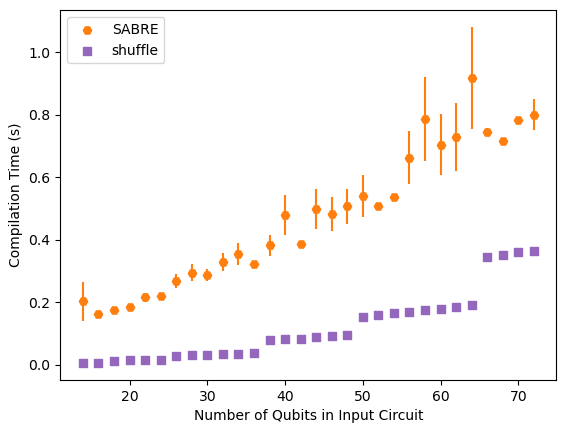}
    \hfill
    \caption{Timing data measuring the compile time of each compilation tool. 2QAN quickly becomes unmanageable. The bottom graph omits 2QAN and Tket to make the scaling of SABRE and shuffle clear.} 
\end{figure}

In Figure \hyperref[timing]{(10)}, we give the compile time for each tool for the compilation of a 12REG graph. When applied to larger and denser graphs, 2QAN can take days or even weeks of compute time to compile a single circuit. When the actual scheduled time of these quantum circuits would be on the order of microseconds, this scaling is unacceptable. This is why shuffle was proposed. Nonetheless, while there is a scaling guarantee of the full shuffle method it does not ever reach the desired logarithmic overhead in any practical setting. 

\begin{figure*}[!htb]
\label{72}
    \minipage{0.32\textwidth}
        \centering
        \includegraphics[width=\textwidth]{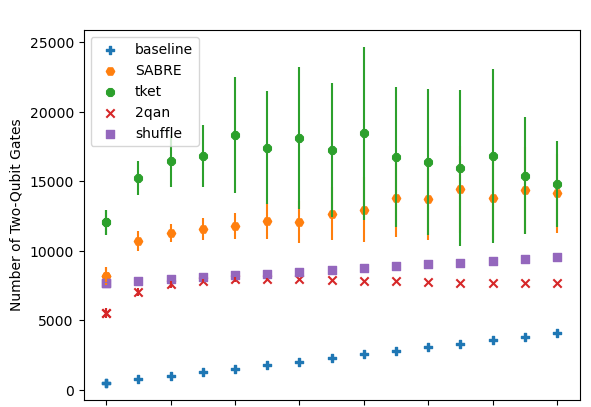}
        \includegraphics[width=\textwidth]{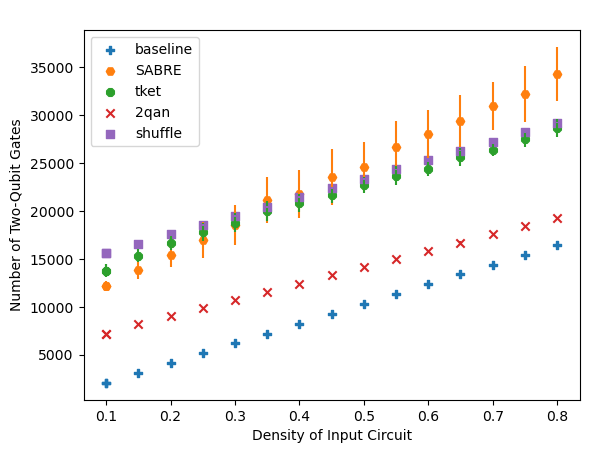}
    \endminipage\hfill
    \minipage{0.32\textwidth}
        \centering
        \includegraphics[width=\textwidth]{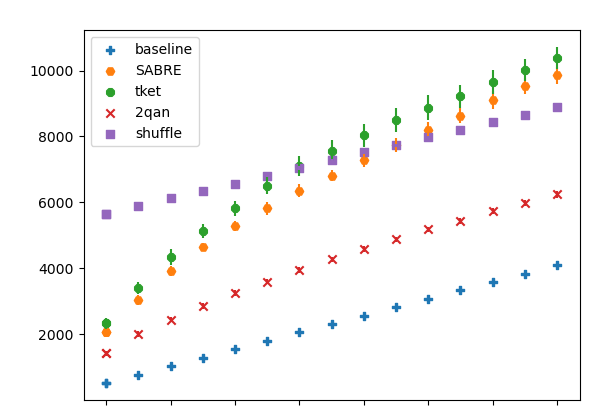}
        \includegraphics[width=\textwidth]{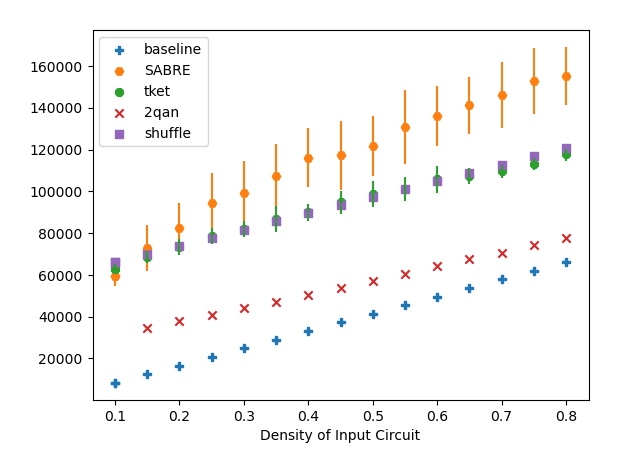}
    \endminipage\hfill
    \minipage{0.32\textwidth}
        \centering
        \includegraphics[width=\textwidth]{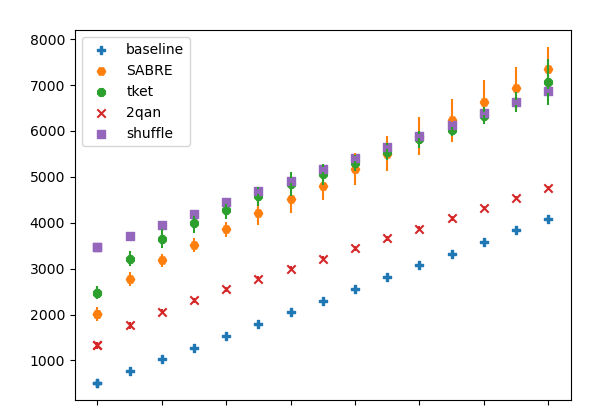}
        \includegraphics[width=\textwidth]{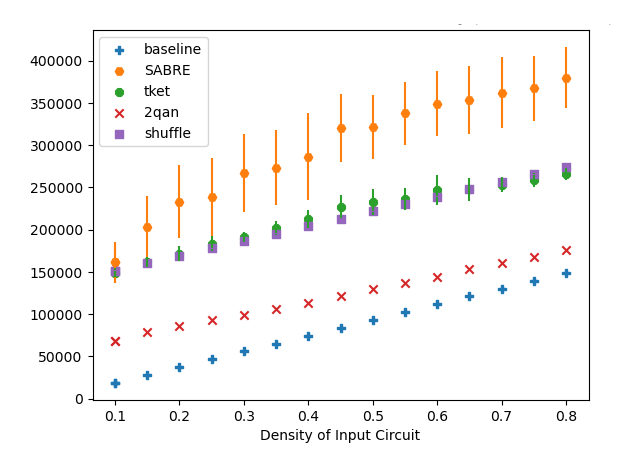}
    \endminipage\hfill
    \caption{How entangling gate count scales with problem instance density for random graphs. From left to right, the top row depicts the line, grid, and busNNN topologies all with circuit width $n=72$. The bottom row depicts the busNNN architecture with $n=144,288,432$ qubits from left to right.} 
\end{figure*}

Aside from taking weeks to compile a single circuit, scaling to larger problem sizes does not have much of an effect on the qualitative behavior we observe at smaller problem instances. One meaningful difference is that SABRE seems to do worse for larger and denser circuits. Along that same vein, we can see that shuffle does become better than both Tket and SABRE past some threshold, as conjectured above. This happens at larger densities for more connected architectures. This fits with the scaling we derived in Section \hyperref[compile]{(IV)}. It is clear that 2QAN consistently outperforms other heuristics as it is taking advantage of optimization invisible to the other compiling methods. So far, we have only discussed hypothetical architectures. The three commercial architectures introduced in Section \hyperref[arch]{(III)} were Google's Sycamore, IBM's heavyhex, and Rigetti's Aspen. We do not derive the swap layers required to apply the shuffle compilation method to these topologies. One simple strategy could be to unwind Rigetti's Aspen and IBM's heavyhex architectures into a line and apply the line swap strategy (with slight correction terms in heavyhex' case). A slightly modified version of the 2D grid swap strategy could be applied to Google's Sycamore. Nonetheless, the three heuristic algorithms' data is compared in Table \hyperref[tab4]{(IV)}. The layered Sycamore chip is also included.

\begin{table}
\label{tab4}
    \centering
    \begin{tabular}{ |p{2.2cm}||p{1.05cm}|p{1.05cm}|p{1.05cm}|p{1.05cm}|p{1.05cm}}
        \hline
        \multicolumn{5}{|c|}{Google's Sycamore} \\
        \hline
        \centering
        \textit{Entangling Gates} & 18q & 36q & 54q & 72q \\
        \hline
        \centering
        SABRE & 174$\pm$13 & 460$\pm$32 & 833$\pm$44 & 1286$\pm$69 \\
        \hline 
        \centering
        Tket & 179$\pm$14 & 489$\pm$34 & 884$\pm$55 & 1435$\pm$84 \\
        \hline 
        \centering
        2QAN & 133$\pm$13 & 349$\pm$27 & 614$\pm$43 & 966$\pm$56 \\
        \hline
        \centering
        \textit{Scaled Sched Time} & & & & \\
        \hline
        \centering
        SABRE & 944$\pm$134 & 1827$\pm$260 & 2838$\pm$339 & 3991$\pm$509 \\
        \hline 
        \centering
        Tket & 1124$\pm$147 & 2309$\pm$312 & 3617$\pm$364 & 5207$\pm$586 \\
        \hline 
        \centering
        2QAN & 480$\pm$79 & 786$\pm$133 & 1095$\pm$185 & 1435$\pm$259 \\
        \hline
        \multicolumn{5}{|c|}{Layered Sycamore} \\
        \hline
        \centering
        \textit{Entangling Gates} & & & & \\
        \hline
        \centering
        SABRE & 157$\pm$11 & 394$\pm$24 & 687$\pm$35 & 1040$\pm$48 \\
        \hline 
        \centering
        Tket & 142$\pm$10 & 369$\pm$25 & 659$\pm$39 & 1002$\pm$60 \\
        \hline 
        \centering
        2QAN & 121$\pm$11 & 295$\pm$21 & 508$\pm$31 & 766$\pm$42 \\
        \hline
        \centering
        \textit{Scaled Sched Time} & & & & \\
        \hline
        \centering
        SABRE & 868$\pm$117 & 1603$\pm$225 & 2417$\pm$298 & 3267$\pm$424 \\
        \hline 
        \centering
        Tket & 918$\pm$114 & 1861$\pm$231 & 2793$\pm$289 & 3852$\pm$420 \\
        \hline 
        \centering
        2QAN & 428$\pm$79 & 653$\pm$115 & 879$\pm$140 & 1066$\pm$171 \\
        \hline
        \multicolumn{5}{|c|}{IBM's Heavy Hex} \\
        \hline
        \centering
        \textit{Entangling Gates} & & & & \\
        \hline
        \centering
        SABRE & 224$\pm$21 & 606$\pm$46 & 1095$\pm$62 & 1658$\pm$91 \\
        \hline 
        \centering
        Tket & 244$\pm$26 & 655$\pm$54 & 1188$\pm$82 & 1819$\pm$115 \\
        \hline 
        \centering
        2QAN & 172$\pm$19 & 450$\pm$38 & 817$\pm$58 & 1259$\pm$87 \\
        \hline
        \centering
        \textit{Scaled Sched Time} & & & & \\
        \hline
        \centering
        SABRE & 1259$\pm$183 & 2634$\pm$358 & 4017$\pm$472 & 5603$\pm$685 \\
        \hline 
        \centering
        Tket & 1395$\pm$170 & 2932$\pm$349 & 4756$\pm$508 & 6560$\pm$812 \\
        \hline 
        \centering
        2QAN & 603$\pm$118 & 1055$\pm$185 & 1478$\pm$245 & 1932$\pm$322 \\
        \hline
        \multicolumn{5}{|c|}{Rigetti's Aspen} \\
        \hline
        \centering
        \textit{Entangling Gates} & & & & \\
        \hline
        \centering
        SABRE & 206$\pm$22 & 578$\pm$45 & 1072$\pm$64 & 1681$\pm$90 \\
        \hline 
        \centering
        Tket & 202$\pm$15 & 635$\pm$48 & 1205$\pm$95 & 1967$\pm$106 \\
        \hline 
        \centering
        2QAN & 149$\pm$17 & 433$\pm$38 & 830$\pm$65 & 1307$\pm$72 \\
        \hline
        \centering
        \textit{Scaled Sched Time} & & & & \\
        \hline
        \centering
        SABRE & 1048$\pm$164 & 2198$\pm$307 & 3534$\pm$403 & 5052$\pm$659 \\
        \hline 
        \centering
        Tket & 1152$\pm$136 & 2720$\pm$295 & 4470$\pm$415 & 6810$\pm$631 \\
        \hline 
        \centering
        2QAN & 529$\pm$96 & 1025$\pm$185 & 1696$\pm$292 & 2400$\pm$392 \\
        \hline
    \end{tabular}
    \vspace{10pt}
    \caption{Data averaged over 100 randomized problem instances of Watts Strogatz (WS). Values are presented as average $\pm$ standard deviation.}
\end{table}

We now turn to the performance of the circuits utilizing the different coupling maps. In Figure \hyperref[BASK]{(12)} we show a comparison 
across the different architectures.
A few striking observations may be made. For one, the ordering of lowest entangling gate count or smallest scaled scheduled time is not constant across problem instances. In fact, this list shifts significantly. The line architecture does not perform well in 12REG but does in Watts Strogatz and Sherrington-Kirkpatrick. This is interesting but may be understood from the fact that both the heavy hex and Aspen topologies are not much more than lines themselves with a few added connections. So, when the graph is dense enough the line architecture performs comparable to both commercial architectures. 

Another interesting takeaway is that added complexity  associated with non-planar coupling maps is not always worthwhile. For a fully connected graph the busNNN architecture outperforms all other architectures by a wide margin. However, for that same model the layered Sycamore chip and a generic 2D grid provide practically identical overheads. For the other two problem instances shown in Figure \hyperref[BASK]{(12)} the gap between busNNN and other architectures exists but is not significant in the regime of circuit widths shown. In these cases, using a simple 2D grid (which is easier to fabricate) would be acceptable. The main takeaway from these observations is that making blanket statements about one architecture's utility versus another means very little without clarifying what problems one will be working with. Even in the context of 2-local Hamiltonians the variation between problem instances is nontrivial.

\begin{figure*}[!htb]
\label{BASK}
    \minipage{0.32\textwidth}
        \centering
        \includegraphics[width=0.95\textwidth]{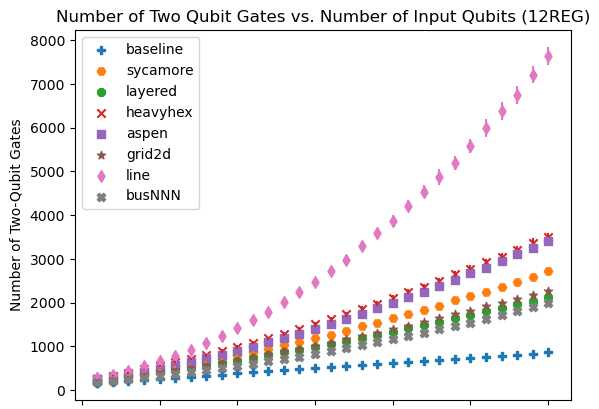}
        \includegraphics[width=0.95\textwidth]{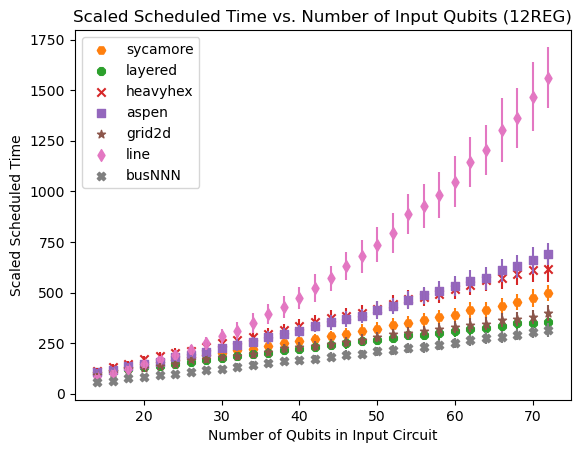}
    \endminipage\hfill
    \minipage{0.32\textwidth}
        \centering
        \includegraphics[width=0.95\textwidth]{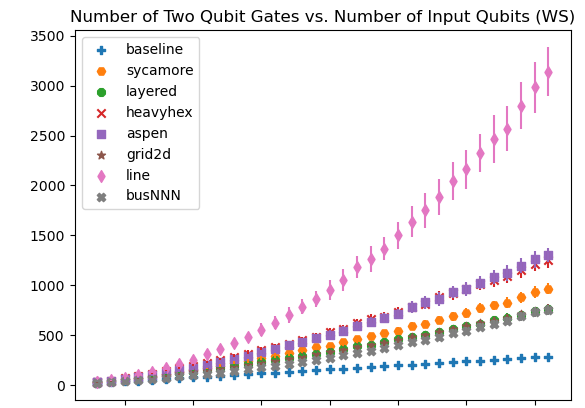}
        \includegraphics[width=0.95\textwidth]{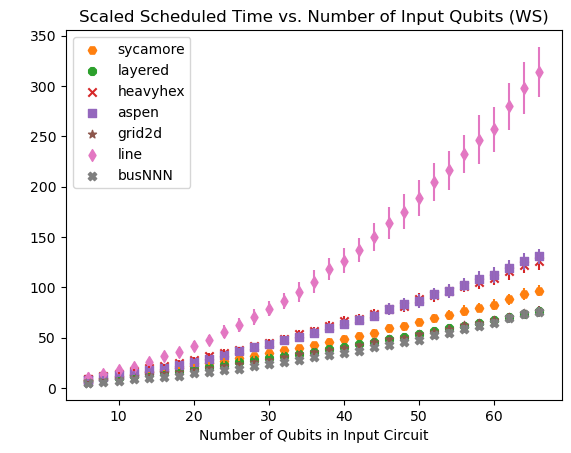}
    \endminipage\hfill
    \minipage{0.32\textwidth}
        \centering
        \includegraphics[width=0.95\textwidth]{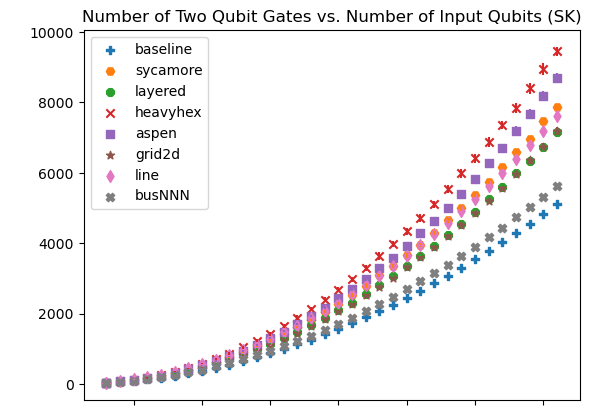}
        \includegraphics[width=0.95\textwidth]{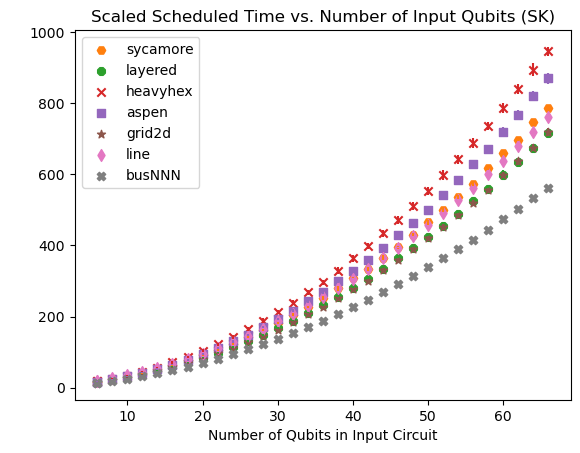}
    \endminipage\hfill
    \caption{Comparison of all architectures, using the best compiler for each. From left to right we have 12REG, WS, and SK problem instances. The top row has two qubit gates as the chosen metric and the bottom row is using scaled scheduled time.} 
\end{figure*}
\section{Conclusion}
\label{conclusion}

It is clear that any design decisions considering coupling maps of superconducting NISQ architectures must be carefully weighed with data spanning compiling tools and problem instances. We have presented benchmarks of the performance of various such tools across both coupling maps and problem instances we believe to be commercially and scientifically relevant. The newly proposed busNNN topology generally has the best performance in terms of the given metrics: entangling gate count and scaled scheduled time. Similarly, 2QAN consistently outperforms other heuristic compiling algorithms while the deterministic shuffle method of compilation outperforms all heuristics past some threshold connectivity and problem instance density. However, our benchmarking uses specific problem instances and gatesets. Future work can expand upon our choices and test other algorithms beyond 2-local Hamiltonians. This may require further research or even development of algorithms that can be run on NISQ hardware. It is an open question which algorithms make sense to run on near-term devices both in terms of size (width/depth) and competitiveness with classical algorithms.

Future work can also investigate how to bridge the performance gap of heuristic compiling algorithms in terms of our chosen metric and the scaling of deterministic compiling algorithms. Commutative SABRE has the option to shuffle an architecture to some desired mixing coefficient then apply SABRE to the resulting circuit. The impact of initial mapping on this strategy is not well understood. It is still possible to combine deterministic and heuristic algorithms in more clever ways; for example, to take advantage of the hierarchical nature of the busNNN architecture which we specifically designed to have a provably optimal swap strategy. Any progress towards lowering the overhead of compilation tools and their time scaling is a worthwhile pursuit. 

The results presented in this study give evidence for utilizing QAOA (or recently optimized variations of QAOA) on certain problem graphs, on specific coupling maps compiled by distinct tools. The obvious goal of the work is to inform actual experiments on which choices to make in regards to architectures for superconducting NISQ devices. The methodology and tools developed for this work can be applied to new problem instances and architectures of interest in the future.

\section{Acknowledgment}

We would like to thank Arthur Kurlej for his help with software development as well as previous work related to benchmarking hardware \cite{line}. We thank John Blue for his insightful comments related to the swap strategy derived for the busNNN architecture.

\bibliographystyle{ieeetr}
\bibliography{main}

\begin{thebibliography}{10}

\bibitem{flux}
L.~Ding, M.~Hays, Y.~Sung, B.~Kannan, J.~An, A.~Paolo, A.~Karamlou, T.~Hazard,
  K.~Azar, D.~Kim, B.~Niedzielski, A.~Melville, M.~Schwartz, J.~Yoder,
  T.~Orlando, S.~Gustavsson, J.~Grover, K.~Serniak, and W.~Oliver,
  ``High-fidelity, frequency-flexible two-qubit fluxonium gates with a transmon
  coupler,'' {\em Physical Review X}, 2023.

\bibitem{googleSupremacy}
{Google Quantum AI}, ``Quantum supremacy using a programmable superconducting
  processor,'' {\em Nature}, 2019.

\bibitem{IBMUtility}
Y.~Kim, A.~Eddins, S.~Anand, K.~Wei, E.~van~den Berg, S.~Rosenblatt, H.~Nayfeh,
  Y.~Wu, M.~Zaletel, K.~Temme, and A.~Kandala, ``Evidence for the utility of
  quantum computing before fault tolerance,'' {\em Nature}, 2023.

\bibitem{depth}
S.~Herbert, ``On the depth overhead incurred when running quantum algorithms on
  near-term quantum computers with limited qubit connectivity,'' {\em Quantum
  Information and Computation}, vol.~20, pp.~787--806, 08 2020.

\bibitem{networkx}
A.~Hagberg, P.~Swart, and D.~S~Chult, ``Exploring network structure, dynamics,
  and function using networkx,'' tech. rep., Los Alamos National Lab.(LANL),
  Los Alamos, NM (United States), 2008.

\bibitem{aspen}
A.~Bestwick, ``Introducing the ankaa-1 system - rigetti's most sophisticated
  chip architecture unlocks a promising path to narrow quantum advantage,''
  2023.

\bibitem{heavyhex}
{IBM Quantum}, ``The ibm quantum heavy hex lattice,'' 2021.

\bibitem{bus}
R.~Stassi, M.~Cirio, and F.~Nori, ``Scalable quantum computer with
  superconducting circuits in the ultrastrong coupling regime,'' {\em npj
  Quantum Information}, 2020.

\bibitem{busComponent}
B.~marinelli, J.~Luo, H.~Ren, B.~Niedzielski, D.~Kim, R.~Das, M.~Schwartz,
  D.~Santiago, and I.~Siddiqi, ``Dynamically reconfigurable photon exchange in
  a superconducting quantum processor,'' {\em PrePrint arXiv:2303.03507}, 2023.

\bibitem{IsingNP}
A.~Lucas, ``Ising formulations of many np problems,'' {\em Frontiers in
  Physics}, vol.~2, 2014.

\bibitem{tutorialQUBO}
F.~Glover, G.~Kochenberger, R.~Hennig, and Y.~Du, ``Quantum bridge analytics i:
  a tutorial on formulating and using qubo models,'' {\em Annals of
  Operatorions Research}, 2022.

\bibitem{NPQUBO}
B.~Lodewijks, ``Mapping np-hard and np-complete optimization problems to
  quadratic unconstrained binary optimization problems,'' {\em PrePrint
  arXiv:1911.08043v4}, 2020.

\bibitem{QAOA}
E.~Farhi, J.~Goldstone, and S.~Gutmann, ``A quantum approximate optimization
  algorithm,'' {\em PrePrint arXiv:1411.4028}, 2014.

\bibitem{boundedQAOA}
E.~Farhi, J.~Goldstone, and S.~Gutmann, ``A quantum approximate optimization
  algorithm applied to a bounded occurrence constraint problem,'' {\em PrePrint
  arXiv:1412.6062v2}, 2015.

\bibitem{QAOAsupremacy}
E.~Fahri and A.~Harrow, ``Quantum supremacy through the quantum approximate
  optimization algorithm,'' {\em PrePrint arXiv:1602.07674v2}, 2019.

\bibitem{combinatorialOptimization}
B.~Korte and J.~Vygen, {\em Combinatorial Optimization Theory and Algorithms}.
\newblock Springer-Verlag Berlin Heidelberg, 2000.

\bibitem{relaxRound}
M.~Dupont and B.~Sundar, ``Quantum relax-and-round algorithm for combinatorial
  optimization,'' {\em PrePrint arXiv:2307.05821}, 2023.

\bibitem{recursiveQAOA}
J.~Finžgar, A.~Kerschbaumer, M.~Schuetz, C.~Mendl, and H.~Katzgraber,
  ``Quantum-informed recursive optimization algorithms,'' {\em PrePrint
  arXiv:2308.13607}, 2023.

\bibitem{SABRE}
G.~Li, Y.~Ding, and Y.~Xie, ``Tackling the qubit mapping problem for nisq-era
  quantum devices,'' in {\em Proceedings for the International Conference on
  Architectural Support for Programming Languages and Operating Systems},
  pp.~1001--1014, Association for Computing Machinery, 2019.

\bibitem{tket}
A.~Cowtan, S.~Dilkes, R.~Duncan, A.~Krajenbrink, W.~Simmons, and S.~Sivarajah,
  ``On the qubit routing problem,'' in {\em Leibniz International Proceedings
  in Informatics}, 2019.

\bibitem{2qan}
L.~Lao and D.~Browne, ``2qan: a quantum compiler for 2-local qubit hamiltonian
  simulation algorithms,'' in {\em Proceedings for International Symposium for
  Computer Architecture}, Association for Computing Machinery, 2022.

\bibitem{peephole}
J.~Liu, L.~Bello, and H.~Zhou, ``Relaxed peephole optimization: a novel
  compiler optimization for quantum circuits,'' in {\em International Symposium
  on Code Generation and Optimization}, IEEE Association for Computing
  Machinery, 2021.

\bibitem{wolframAdjacency}
L.~Sauras-Altuzarra and E.~Weisstein, ``Adjacency matrix,'' 2023.

\bibitem{tabu}
M.~Gendreau and J.-Y. Potvin, {\em Tabu Search}, pp.~165--186.
\newblock Boston, MA: Springer US, 2005.

\bibitem{maaps}
Y.~Jin, J.~Luo, L.~Fong, Y.~Chen, A.~Hayes, C.~Zhang, F.~Hua, and E.~Zhang, ``A
  structured method for compiling and optimizing qaoa circuits in quantum
  computing,'' {\em PrePrint arXiv:2112.06143}, 2022.

\bibitem{commtSABRE}
J.~Weidenfeller, L.~Valor, J.~Gacon, C.~Tornow, L.~Bello, S.~Woerner, and
  D.~Egger, ``Scaling of the quantum approximate optimization algorithm on
  superconducting qubit based hardware,'' {\em Quantum}, 2022.

\bibitem{mikeIke}
M.~Nielsen and I.~Chuang, {\em Quantum Computation and Quantum Information}.
\newblock Cambridge University Press, 2000.

\bibitem{grid}
A.~Reuther, J.~Kepner, C.~Byun, S.~Samsi, W.~Arcand, D.~Bestor, B.~Bergeron,
  V.~Gadepally, M.~Houle, M.~Hubbell, M.~Jones, A.~Klein, L.~Milechin,
  J.~Mullen, A.~Prout, A.~Rosa, C.~Yee, and P.~Michaleas, ``Interactive
  supercomputing on 40,000 cores for machine learning and data analysis,'' in
  {\em 2018 IEEE High Performance extreme Computing Conference (HPEC)},
  pp.~1--6, IEEE, 2018.

\bibitem{line}
A.~Kurlej, S.~Alterman, and K.~Obenland, ``Benchmarking and analysis of noisy
  intermediate-scale trapped ion quantum computing architectures,'' {\em IEEE
  International Conference on Quantum Computing and Engineering}, 2022.

\end{thebibliography}

\end{document}